\documentclass[journal,twocolumn,10pt]{IEEEtran}
\usepackage{amsmath,epsfig} % spconf,

\usepackage{floatflt} % text flows around figures
\usepackage{paralist} % compact and in-paragraph lists
\usepackage{algorithm} % package for pseudo code
\usepackage{graphicx}
\usepackage{amsthm}
\usepackage{amssymb}
\usepackage{listings}    %输入代码必备
\usepackage{xcolor}
\usepackage{hyperref}
\usepackage{indentfirst}
\usepackage{booktabs}
\usepackage{subfigure}
\usepackage{caption}
\usepackage{calligra}
\usepackage{bm}
\usepackage{fancyhdr}
\usepackage{float}
\usepackage{pifont}% http://ctan.org/pkg/pifont
\usepackage{colortbl} % color for tables
\usepackage{amsmath}
\usepackage{MnSymbol}%
\usepackage{amsbsy}

\theoremstyle{definition}

\newcommand{\mypar}[1]{{\bf #1.}}

\newtheorem{myThm}{Theorem}

\newcommand{\R}{\ensuremath{\mathbb{R}}}

\DeclareMathOperator{\Id}{I}

\def\x{\mathbf{x}}
\def\y{\mathbf{y}}
\def\h{\mathbf{h}}

\def\p{\mathbf{p}}

\def\ee{\mathbf{e}}
\def\s{\mathbf{s}}
\def\t{\mathbf{t}}
\def\tt{\mathbf{t}}
\def\vv{\mathbf{v}}

\def\z{\mathbf{z}}

\def\V{\mathcal{V}}
\def\E{\mathcal{E}}
\def\M{\mathcal{M}}

\def\Aa{\mathbb{A}}
\def\Bb{\mathbb{B}}
\def\Cc{\mathbb{C}}
\def\Dd{\mathbb{D}}
\def\Ee{\mathbb{E}}

\def\Hh{\mathbb{H}}

\DeclareMathOperator{\Adj}{A}

\DeclareMathOperator{\CC}{C}
\DeclareMathOperator{\D}{D}

\DeclareMathOperator{\HH}{H}

\DeclareMathOperator{\Pj}{P}
\DeclareMathOperator{\Q}{Q}

\DeclareMathOperator{\J}{J}

\DeclareMathOperator{\Vm}{V}

\DeclareMathOperator{\EE}{E}
\DeclareMathOperator{\T}{T}
\DeclareMathOperator{\X}{X}
\DeclareMathOperator{\Y}{Y}
\def\Xw{\widehat{\X}}

\DeclareMathOperator{\Z}{Z}

\DeclareMathOperator{\Ss}{S}

\title{Graph Unrolling Networks: Interpretable Neural Networks for Graph Signal Denoising}
\author{Siheng Chen~\IEEEmembership{ Member,~IEEE}, Yonina C. Eldar~\IEEEmembership{Fellow, IEEE}, Lingxiao Zhao\thanks{S. Chen is with Mitsubishi Electric Research Laboratories (MERL), Cambridge, MA, USA. Email: schen@merl.com. Y. C. Eldar is with 
Math and CS department at Weizmann Institute of Science, Rehovot, Israel. Email: yonina.eldar@weizmann.ac.il. L. Zhao is with Heinz College at Carnegie Mellon University, Pittrsbugh, PA, USA. Email: lingxiao@cmu.edu. }}
\begin{document}

\maketitle
%\tableofcontents

\begin{abstract}
We propose an interpretable graph neural network framework to denoise single or multiple noisy graph signals. The proposed~\emph{graph unrolling networks} expand algorithm unrolling to the graph domain and provide an interpretation of the architecture design from a signal processing perspective. We unroll an iterative denoising algorithm by mapping each iteration into a single network layer where the feed-forward process is equivalent to iteratively denoising graph signals. We train the graph unrolling networks through unsupervised learning, where the input noisy graph signals are used to supervise the networks. By leveraging the learning ability of neural networks, we adaptively capture appropriate priors from input noisy graph signals, instead of manually choosing signal priors. A core component of graph unrolling networks is the~\emph{edge-weight-sharing graph convolution operation}, which parameterizes each edge weight by a trainable kernel function where the trainable  parameters are shared by all the edges. The proposed convolution is permutation-equivariant and can flexibly adjust the edge weights to various graph signals. We then consider two special cases of this class of networks, graph unrolling sparse coding (GUSC) and graph unrolling trend filtering (GUTF), by unrolling sparse coding and trend filtering, respectively. To validate the proposed methods, we conduct extensive experiments on both real-world datasets and simulated datasets, and demonstrate that our methods have smaller denoising errors than conventional denoising algorithms and state-of-the-art graph neural networks. For denoising a single smooth graph signal, the normalized mean square error of the proposed networks is around $40\%$ and $60\%$ lower than that of graph Laplacian denoising and graph wavelets, respectively.
\end{abstract}


\section{Introduction}
\label{sec:intro}
Data today is often generated from a diverse sources, including social, citation, biological, and physical infrastructure~\cite{OrtegaFKMV:18}. Unlike time-series signals or images, such signals possess complex and irregular structures, which can be modeled as graphs. Analyzing graph signals requires dealing with the underlying irregular relationships. Graph signal processing generalizes the classical signal processing toolbox to the graph domain and provides a series of techniques to process graph signals~\cite{OrtegaFKMV:18}, including graph-based transformations~\cite{HammondVG:11,ShumanFV:16}, sampling and recovery of graph signals~\cite{ChenVSK:15,AnisAO:15} and graph topology learning~\cite{DongTRF:19}. Graph neural networks provide a powerful framework to learn from graph signals with graphs as induced biases~\cite{BronsteinBLSV:17}. Permeating the benefits of deep learning to the graph domain, graph convolutional networks and variants have attained remarkable success in social network analysis~\cite{KipfW:17}, 3D point cloud processing~\cite{WangSLSBS:19}, quantum chemistry~\cite{GilmerSRVD:17} and computer vision~\cite{LiCZZWT:20}.

In this work, we consider denoising graph signals~\cite{ShumanNFOV:13,ChenSMK:14a}. In classical signal processing, signal denoising is one of the most ubiquitous tasks~\cite{VetterliKG:12}. To handle graph signals, there are two mainstream approaches: graph-regularization-based optimization and graph dictionary design. The optimization approach usually introduces a graph-regularization term that promotes certain properties of the graph signal and solves a regularized optimization problem to obtain a denoised solution~\cite{ChenSMK:14a,WangSST:16,VarmaLKC:20}.  In image denoising, the total variation, which captures the integral of the absolute gradient of the image~\cite{RudinOF:92,Chambolle:04}, is often used. Minimizing the total variation of an image helps remove unwanted noise while preserving important details, such as edges and contours. In graph signal denoising, the regularizers are often chosen to relate to the graph properties. For example, a popular choice is a quadratic form of the graph Laplacian, which captures the second-order difference of a graph signal, corresponding to a graph smoothness prior~\cite{Chung:96,ShumanNFOV:13}. Graph total variation captures the sparsity of the first-order difference of a graph signal, reflecting piecewise-constant prior~\cite{ChenSMK:15,WangSST:16}. Some variations of graph total variation can capture more complicated priors and lead to fast algorithms~\cite{VarmaLKC:20}.

In comparison, the graph-dictionary approach aims to reconstruct graph signals through a predesigned graph dictionary, such as graph wavelets~\cite{HammondVG:11}, windowed graph Fourier transforms~\cite{ShafipourKM:19}, and graph frames~\cite{ShumanWHV:15}. These dictionaries are essentially variants of graph filters. A reconstructed graph signal consists of a sparse combination of elementary graph signals in a graph dictionary.  The combination coefficients can be obtained through sparse coding algorithms, such as matching pursuit and basis pursuit~\cite{EldarK:12,Eldar:15}. A fast implementation of graph dictionary are graph filter banks, which use a series of band-pass graph filters that expands the input graph signal into multiple subband components~\cite{NarangO:12,ShumanWHV:15,ShumanFV:16,SakiyamaWTO:19}. By adjusting the component in each subband, a graph filter bank can flexibly modify a graph signal and suppress noises, especially in the high-frequency band. 

A fundamental challenge for both denoising approaches is that we may not know an appropriate prior on the noiseless graph signals in practice. It is then hard to either choose an appropriate graph-regularization term or design an appropriate graph dictionary. Furthermore, some graph priors are too complicated to be explicitly and precisely described in mathematical terms or may lead to complicated and computationally intensive algorithms.

To solve this issue, it is desirable to learn an appropriate prior from given graph signals; in other words, the denoising algorithm should have sufficient feasibility to learn from and adapt to arbitrary signal priors. Deep neural networks have demonstrated strong power in learning abstract, yet effective features from a huge amount of data~\cite{Goodfellow:2016}. As the extension of neural networks to the graph domain, graph neural networks have also received a lot of attention and achieved significant success in social network analysis and geometric data analysis~\cite{KipfW:17, Battaglia:18}. One mainstream graph neural network architecture is the graph convolutional network (GCN), which relies on a layered  architecture that consists of trainable graph convolution operations, followed by pointwise nonlinear functions~\cite{BrunaZSL:13,DefferrardBV:16,KipfW:17}. Some variants include  graph attention networks~\cite{VelivckovicCCRLB:18}, deep graph infomax~\cite{VelickovicFHLBH:19}, simple graph convolution~\cite{WuSZFYW:19}, and the graph U-net~\cite{GaoJ:19}. These GCN-based models have shown remarkable success in graph-based semi-supervised learning~\cite{KipfW:17,VelivckovicCCRLB:18} and graph classification tasks~\cite{ZhangCNC:18}. However, these neural network architectures are typically designed through trial and error. It is thus hard to explain the design rationale and further improve the architectures~\cite{SolomonCZYHLSE:18}. 

In this work, we leverage the powerful learning ability of graph neural networks and combine them with interpretablity based on a signal processing perspective. Furthermore, most graph neural networks are developed for supervised-learning tasks, such node classification~\cite{KipfW:17}, link prediction~\cite{ZhangC:18} and graph classification~\cite{GaoJ:19}. Those tasks require a large number of ground-truth labels, which is expensive to obtain. Here we consider an unsupervised-learning setting, where the networks have to learn from a few noisy graph signals and the ground-truth noiseless graph signals are unknown. Through unsupervised learning, we demonstrate the generalization ability of the proposed graph neural networks.

Our goal is to develop a framework for graph network denoising by combining the advantages of both conventional graph signal denoising algorithms and graph neural networks. On the one hand, we follow the iterative procedures of conventional denoising algorithms, which provides interpretability and explicit graph regularization; on the other hand, we parameterize a few mathematically-designed operations through neural networks and train the parameters, which provides flexibility and learning ability. We bridge between conventional graph signal denoising algorithms and graph neural networks by using the powerful framework of algorithm unrolling~\cite{MongaLE:19}. It provides a concrete and systematic connection between iterative algorithms in signal processing and deep neural networks, and paves the way to developing interpretable network architectures. The related unrolling techniques have been successfully applied in many problem areas, such as sparse coding~\cite{GregorL:10}, ultrasound signal processing~\cite{SolomonCZYHLSE:18},
image deblurring~\cite{LiTME:19,LiTGME:20} and image denoising~\cite{ShabiliMB:20}.

In this work, we expand algorithm unrolling to the graph domain. We first propose a general iterative algorithm for graph signal denoising and then transform it to a graph neural network through algorithm unrolling, where each iteration is mapped to a network layer. Compared to conventional denoising algorithms~\cite{ChenSMK:14a}, the proposed graph unrolling network is able to learn a variety of priors from given graph signals by leveraging deep neural networks. Compared to many other graph neural networks~\cite{KipfW:17}, the proposed graph unrolling network is interpretable by following analytical iterative steps. To train graph unrolling networks, we use single or multiple noisy graph signals and minimize the difference between the original input, which is the noisy graph  signal, and the network output, which is the denoised graph signal; in other words, the input  noisy measurements are  used  to  supervise the neural network training. Surprisingly, even when we train until convergence, the output does not overfit the noisy input in most cases. The intuition is that the proposed operations and architectures of graph unrolling networks carry implicit graph regularization and thus avoid overfitting.

A core component in the proposed graph unrolling networks is the edge-weight-sharing graph convolution operation. The proposed graph convolution parameterizes edge weights by a trainable kernel function, which maps a pair of vertex coordinates in the graph spectral domain to an edge weight. Since the trainable parameters in the kernel function are shared by all the edges, this graph convolution has a weight-sharing property. We also show that it is equivariant to the permutation of vertices. Our convolution is different from conventional graph filtering in graph signal processing~\cite{ChenSMK:14a}, as it includes trainable parameters that can transform a graph signal from the original graph vertex domain to a high-dimensional feature domain. It is also different from  many trainable graph convolutions~\cite{KipfW:17}, since it adjusts edge weights according to the input graph signals during training, which makes it flexible to capture complicated signal priors.

Based on the graph unrolling network framework, we further propose two specific architectures by unrolling graph sparse coding and graph trend filtering. These two are typical denoising algorithms based on graph dictionary design and graph-regularization-based optimization, respectively. We consider both networks to demonstrate the generalization of the proposed framework. Both methods are special cases of the  general  iterative  algorithm  for graph signal denoising.

To validate the empirical performance of the proposed method, we conduct a series of experiments on both simulated datasets and real-world datasets with Gaussian noises, mixture noises and Bernoulli noises. We find that graph unrolling networks consistently achieve better denoising performances than conventional graph signal denoising algorithms and state-of-the-art graph neural networks on various types of graph signals and noise models. We also find that even for denoising a single smooth graph signal, the proposed graph unrolling networks are around $40\%$ and $60\%$ better than graph Laplacian denoising~\cite{ShumanNFOV:13} and graph wavelets~\cite{HammondVG:11}, respectively. This demonstrates that the unrolling approach allows to obtain improved results over existing methods even using a single training point.

The main contributions of this work include: 
\begin{itemize}
\item We propose interpretable graph unrolling networks by unrolling a general iterative algorithm for graph signal denoising in an unsupervised-learning setting;

\item We propose a trainable edge-weight-sharing graph convolution whose trainable parameters are shared across all the edges. It is also equivariant to the permutation of vertices;

\item We propose two specific network architectures under the umbrella of graph unrolling networks: graph unrolling sparse coding and graph unrolling trend filtering; and

\item  We conduct experiments on both simulated and real-world data to validate that the proposed denoising methods significantly outperform both conventional graph signal denoising methods and state-of-the-art graph neural networks on various types of graph signals and noise models. The proposed networks work best even for a single training sample.
\end{itemize}

The rest of the paper is organized as follows: Section~\ref{sec:formulation} formulates the graph signal denoising problem and revisits a few classical denoising methods. Section~\ref{sec:graph_convolution} proposes an edge-weight-sharing graph convolution operation, which is a core operation in the proposed network. Section~\ref{sec:GUN} describes the general framework of graph unrolling networks and provides two specific architectures: graph sparse coding and graph trend filtering. Experiments validating the advantages of our methods are provided in Section~\ref{sec:experiments}.

\section{Problem formulation}
\label{sec:formulation}
In this section, we mathematically formulate the task of graph signal denoising and review a few classical denoising methods, which lay the foundation for the proposed methods. 

We consider a graph $G = (\V, \E, \Adj)$, where $\V = \{v_n\}_{n=1}^{N}$ is the set of \emph{vertices}, $\E = \{e_m\}_{m=1}^{M}$ is the set of undirected \emph{edges}, and $\Adj \in \R^{N \times N}$ is the graph adjacency matrix, representing connections between vertices. The weight $\Adj_{i,j}$ of an edge from the $i$th to the $j$th vertex characterizes the relation, such as similarity or dependency, between the corresponding signal values. Using the graph representation $G$, a \emph{graph signal} is defined as a map that  assigns a signal coefficient $x_n \in \R$ to the vertex $v_n$. A graph signal can be written as a length-$N$ vector defined by
\begin{equation*}
\label{eq:graph_signal}
  \x \ = \ \begin{bmatrix}
 x_1 & x_2 & \ldots & x_{N}
\end{bmatrix}^T,
\end{equation*}
where the $n$th vector element $x_n$ is indexed by the vertex $v_n$. 

The multiplication between the graph adjacency matrix and a graph signal, $\Adj\x$, replaces the signal coefficient at each vertex with the weighted linear combination of the signal coefficients at the corresponding neighbors. In other words, the graph adjacency matrix enables the value at each vertex shift to its neighbors; we thus call it a~\emph{graph shift operator}~\cite{SandryhailaM:13}.  In order for the output norm not to increase after graph shifting, we normalize the graph adjacency matrix, $\Adj^{\rm norm} = \Adj / |\lambda_{\rm max}(\Adj)|$, where $\lambda_{\rm max}(\Adj)$ denotes the eigenvalue of $\Adj$ with the largest magnitude. The normalized matrix has spectral
norm $\left\|\Adj^{\rm norm}\right\|_2 = 1$. In this paper, we assume that the all graph shift operators are normalized, that is $\Adj = \Adj^{\rm norm}$. This property will be used in Section~\ref{sec:graph_convolution_theory}.

Assume that we are given a length-$N$ noisy measurement of a graph signal
\begin{equation}
\label{eq:noise_model}
\t = \x + \ee,
\end{equation} 
where $\x$ is the noiseless graph signal and $\ee$ is noise. The goal of graph signal denoising is to recover $\x$ from $\t$ by removing the noise. We can further extend this setting to multiple graph signals. Consider $K$ measurements in a $N \times K$ matrix,
$
  \T \ = \ \begin{bmatrix} \tt^{(1)} & \tt^{(2)} &\ldots & \tt^{(K)} \end{bmatrix} \ = \ \X + \EE,
$
where $\X = \begin{bmatrix} \x^{(1)} & \x^{(2)} & \ldots & \x^{(k)} & \ldots & \x^{(K)} \end{bmatrix}$ is a matrix of $K$ noiseless graph signals, and $\EE$ is a $N\times K$ matrix that contains independent and identically distributed random noises. We thus aim to recover $\X$ from $\T$ by removing the noise $\EE$.

Without any prior information on the noiseless graph signals, it is impossible to split noises from the measurements. Possible priors include sparsity, graph smoothness and graph piecewise-smoothness~\cite{ChenVSK:16}. Here we consider a general graph signal model in which the graph signal is generated through graph filtering over vertices; that is,
\begin{eqnarray*}
 \x \ =  \ \h *_v \s = \sum_{\ell=1}^L h_{\ell} \Adj^\ell \s, 
\end{eqnarray*}
where $\s \in \R^N$ is a base graph signal, which may not have any graph-related properties, $*_v $ indicates a convolution on the graph vertex domain and $\h = \begin{bmatrix}
h_{1} & h_{2} & \cdots & h_{L}  \end{bmatrix}^T \in \R^L$ are the predesigned and fixed filter coefficients with $L$ the filter length. Here we consider a typical design of a graph filter, which is a polynomial of the graph shift; see detailed discussion in Section~\ref{sec:graph_convolution}. The graph filtering process modifies a given base graph signal according to certain patterns of the graph and explicitly regularizes the output.

Based on this graph signal model, we can remove noises by solving the following optimization problem:
\begin{eqnarray}
\label{eq:graph_signal_denoising_optimization}
  && \min_{\s \in \R^N} \frac{1}{2}||\t - \x ||_2^2 +  u( \Pj \x) + r( \Q \s),
	\\ \nonumber
	&& {\rm subject~to~~} \x \ =  \  \h *_v \s,
\end{eqnarray}
where $u(\cdot), r(\cdot) \in \R$ are additional regularization terms on $\s$ and $\x$ respectively and $\Pj$ and $\Q$ are two matrices. The denoised graph signal is then given by $h *_v \s$. 

To connect the general graph signal denoising problem~\eqref{eq:graph_signal_denoising_optimization} to  previous works, we present several special cases of~\eqref{eq:graph_signal_denoising_optimization}.
\subsubsection{Graph sparse coding}
\label{sec:GSC}
Here we consider reconstructing a noiseless graph signal through graph filtering and regularizing the graph signal to be sparse. The optimization problem of graph sparse coding is
\begin{equation}
\label{eq:graph_sparse_coding}
\min_{\s \in \R^N} \frac{1}{2}||\t - \sum_{\ell=1}^L h_{\ell} \Adj^\ell \s ||_2^2 +  \alpha\left\| \s \right\|_1.
\end{equation}

In this setting, $\x \ = \ \h *_v \s \ = \ \sum_{\ell=1}^L h_{\ell} \Adj^\ell \s$, $u( \cdot) = 0, r(\cdot) = \alpha\left\| \cdot \right\|_1$ and $\Pj = \Q = \Id$.  The term $\left\| \s \right\|_1$ promotes sparsity of the base signal and $\h *_v \s$ allows the sparse signal coefficients to diffuse over the graph. Many denoising algorithms based on graph filter banks and graph dictionary representations are variations of graph sparse coding~\cite{ChenVSK:16}. They design various graph filters to adjust subbbands' responses and use matching pursuit or basis pursuit to solve~\eqref{eq:graph_sparse_coding}.

\subsubsection{Graph Laplacian denoising}
Here we consider using the the second-order difference to regularize a graph signal. The optimization problem of graph Laplacian denoising is
\begin{equation}
\label{eq:graph_laplacian_denoising}
\min_{\x \in \R^N} \frac{1}{2}||\t - \x ||_2^2 +  \alpha \x^T \mathcal{L} \x,
\end{equation}
where $\mathcal{L} = \D - \Adj \in \R^{N \times N}$ is
the graph Laplacian matrix with diagonal degree matrix $\D_{i,i} = \sum_{j} \Adj_{i,j}$. In this setting, $\x = \h *_v \s = \s$, $u(\cdot) = \alpha \left\|\cdot\right\|_2^2 , r(\s) = 0$, $\Pj = \mathcal{L}^{\frac{1}{2}}$ and $\Q = \Id$.  Here we  do not consider the effect of graph filtering and directly regularize the graph signal $\x$. The term 
\begin{equation*}
u(\Pj \x) \ = \ \alpha \left\| \mathcal{L}^\frac{1}{2} \x \right\|_2^2 =  \alpha \x^T \mathcal{L} \x = \alpha \sum_{ (i,j) \in \E} \Adj_{i,j} (x_i - x_j)^2,
\end{equation*}
is well known as the quadratic form of the graph Laplacian, which has been widely used in graph-based semi-supervised learning, spectral clustering and graph signal processing~\cite{ZhuLG:03,NgJW:01,OrtegaFKMV:18}. It captures the second-order difference of a graph signal by accumulating the
pairwise differences between signal values associated with adjacent vertices.
When solving~\eqref{eq:graph_laplacian_denoising}, we regularize $\x$ to be smooth~\cite{ZhuLG:03}.

\subsubsection{Graph trend filtering}
\label{sec:GTF}
Here we consider using the first-order difference to regularize a graph signal. The optimization problem is
\begin{equation}
\label{eq:graph_trend_filtering}
\min_{\x \in \R^N} \frac{1}{2}||\t - \x ||_2^2 +  \alpha \left\| \Delta \x \right\|_1,
\end{equation}
where $\Delta$ is a $M \times N$ graph incidence matrix with $M$ the number of edges and $N$ the number of nodes. Each row of $\Delta$ corresponds to an edge. For example, if $e_i$ is an edge that connects the $j$th vertex to the $k$th vertex ($j < k$), the elements of the $i$th row of $\Delta$ are
\begin{equation}
\label{eq:Delta}
 \Delta_{i, \ell} = 
  \left\{ 
    \begin{array}{rl}
      -  \sqrt{ \Adj_{j,k} }, & \ell = j;\\
        \sqrt{ \Adj_{j,k} }, & \ell = k;\\
      0, & \mbox{otherwise}.
  \end{array} \right.
\end{equation}
The graph incident matrix measures the first-order difference and satisfies $\Delta^T \Delta = \mathcal{L}$. In this setting $\x = \h *_v \s = \s$, $u(\cdot) = \alpha \left\| \cdot \right\|_1$, $r(\cdot) = 0$, $\Pj = \Delta$ and $\Q = \Id$. The term
\begin{equation*}
u(\Pj \x) \ = \ \alpha \left\| \Delta \x \right\|_1 =  \alpha \sum_{ (i,j) \in \E} \Adj_{i,j} |x_i - x_j|,
\end{equation*}
is known as the graph total variation and is often used in graph signal denoising.  Similar to the graph-Laplacian regularization, the graph total variation considers pairwise differences. However, it uses the $\ell_1$ norm to promote sparsity of the first-order differences. When solving~\eqref{eq:graph_trend_filtering}, we encourage $\x$ to be piecewise-constant~\cite{WangSST:16,VarmaLKC:20}.

In Section~\ref{sec:GUN}, we solve the general graph signal denoising problem~\eqref{eq:graph_signal_denoising_optimization} through algorithm unrolling and propose a framework for developing graph unrolling networks. Before that, we first propose a core component of graph unrolling networks: adaptive graph convolution.

\section{Edge-weight-sharing graph convolution}
\label{sec:graph_convolution}
In this section, we present a trainable graph convolution operation\footnote{To be consistent, graph convolution in this paper means graph filtering. In fact, graph convolution and graph filtering have two distinct meanings. Graph convolution takes two graph signals as inputs and outputs a third graph signal; graph filtering takes a graph signal and a graph filter as inputs and outputs another graph signal.}, which can be trained in end-to-end learning. The proposed convolution parameterizes each edge weight by a trainable kernel function whose trainable parameters are shared across all the edges; we thus call it~\emph{edge-weight-sharing graph convolution}. It will be used as a building block of the graph unrolling networks in Section~\ref{sec:GUN}.  Here we first revisit the standard graph convolution used in signal processing, and then equip it with trainable parameters, suitable for neural networks.

\subsection{Graph convolution in signal processing}
We first revisit 1D cyclic convolution in conventional discrete signal processing. Let $\x \in \R^N$ be a time-series. The output time-series after cyclic convolution is
\begin{equation*}
\y  \ = \ \h * \x \ = \ \sum_{\ell=1}^L h_{\ell} \CC^{\ell} \x \in \R^N,
\end{equation*}
where $L$ is the length of a filter, $\h = \begin{bmatrix}
h_{1} & h_{2} & \cdots h_{L}
\end{bmatrix}^T \in \R^L$ are the filter coefficients and the cyclic-permutation matrix
\begin{equation*}
\CC  \ = \
\begin{bmatrix}
0 & 0 & \cdots & 0 & 1 \\
1 & 0 & \cdots & 0 & 0 \\ 
0 & 1 & \ddots & 0 & 0 \\ 
\vdots & \vdots & \ddots & \ddots \\ 
0 & 0 & \ddots & 1 & 0 \\ 
\end{bmatrix}
\in \R^{N \times N}
\end{equation*}
is a matrix representation of a directed cyclic graph~\cite{SandryhailaM:13}. It reflects the underlying structure of a finite, periodic discrete time series. All edges are directed and have the same weight 1, reflecting the causality of a time series. A polynomial of the cyclic-permutation matrix, $\sum_{\ell=1}^L h_{\ell} \CC^{\ell}$, is a filter in the time domain. The essence of convolution/filtering is to update a signal coefficient by weighted averaging of the neighboring coefficients.  The neighbors are defined based on the cyclic-permutation matrix $\CC$ where the weights, or the filter coefficients, are shared across the entire signal.

We can use a mathematical analogy to generalize convolution from the time domain to the graph domain; that is, we simply replace the cyclic-permutation matrix $\CC$ by the graph adjacency matrix $\Adj$~\cite{SandryhailaM:13,OrtegaFKMV:18}. Let $\x \in \R^N$ be a a graph signal. The output graph signal after graph convolution is the length $N$ vector
\begin{equation}
\label{eq:graph_convolution}
\y  \ = \ \h *_v \x \ = \ \sum_{\ell=1}^L h_{\ell} \Adj^{\ell} \x.
\end{equation} 
The filter coefficients $\h = \begin{bmatrix}
h_{1} & h_{2} & \cdots & h_{L}\end{bmatrix}$ are usually fixed and designed based on the theories of graph filter banks and graph wavelets~\cite{HammondVG:11}. Some variations also consider the vertex-variant graph convolution~\cite{SegarraMR:16}, where each filter coefficient $h_{\ell}$ is expanded to be a vector, and 
the edge-variant graph convolution~\cite{CoutinoIL:17,IsufiGR:20}, where 
each filter coefficient $h_{\ell}$ is expanded to be a matrix. In this paper, we consider~\eqref{eq:graph_convolution} as our default graph convolution in any non-neural-network-based model. For example, in the general graph signal denoising problem~\eqref{eq:graph_signal_denoising_optimization}, the graph convolution  follows the definition in~\eqref{eq:graph_convolution}. We will use $\nabla \h *_v \x = \h *_v = \sum_{\ell=1}^L h_{\ell} \Adj^{\ell}$ to denote the derivative of $\h *_v \x$, which will be used in Section~\ref{sec:GUN}.

\subsection{Graph convolution in neural networks}
As one of the most successful neural network models, convolution neural networks (CNNs) use a sequence of trainable convolution operations to extract deep features from input data. The convolution in CNNs follows the same spirit of the conventional convolution. . At the same time, it allows feature learning in a high-dimensional space, which is one of the most important characteristics of CNNs~\cite{Goodfellow:2016}. A convolution operator usually carries a large number of trainable parameters and allows for multiple-channel inputs and multiple-channel outputs.

Let $\X = \begin{bmatrix}
\x^{(1)} & \x^{(2)} & \ldots & \x^{(K)} \end{bmatrix} \in \R^{N \times K}$ be a $K$-channel signal; in other words, there are $K$ features at each time stamp. Let a three-mode tensor $\Hh \in \R^{L \times K \times K'}$ be a collection of trainable filter coefficients that takes a $K$-channel signal as input and outputs a $K'$-channel signal. Each convolution layer operates as $\Y = \Hh * \X \in \R^{N \times K'}$, with $k'$th output channel
\begin{equation}
\label{eq:neural_convolution}
\y^{(k')} \ = \ \sum_{\ell=1}^L \sum_{k=1}^{K} \Hh_{\ell,k,k'} \CC^{\ell} \x^{(k)},
\end{equation}
where $\ell$ is the index of the filter length, $k$ is the index of the input channel, and $k'$ is the index of the output channel. 
Each element in $\Hh$ is trainable and is updated in an end-to-end learning process. This is the standard convolution operation that is widely used in many applications, such as speech recognition and computer vision~\cite{Goodfellow:2016}. Comparing to conventional convolution in signal processing,~\eqref{eq:neural_convolution} introduces trainable parameters and the choices of $K, K'$ allow feature learning in high-dimensional spaces.

We can extend trainable convolution to the graph domain by replacing the cyclic-permutation matrix with the graph adjacency matrix. Similar generalizations have been explored in~\cite{GamaMLR:19}. Define $\X = \begin{bmatrix}
\x^{(1)} & \x^{(2)} & \ldots & \x^{(K)} \end{bmatrix} \in \R^{N \times K}$ as a $K$-channel graph signal; in other words, there are $K$ features at each vertex. Let a three-mode tensor $\Hh \in \R^{L \times K \times K'}$ be a collection of trainable graph filter coefficients that takes a $K$-channel graph signal as input and outputs a $K'$-channel graph signal. A trainable graph convolution is 
$$\Y = \Hh *_w \X,$$ with the $k'$th output channel
\begin{equation}
\label{eq:neural_graph_convolution}
\y^{(k')} \ = \ \sum_{\ell = 1}^L \sum_{k=1}^K \Hh_{\ell,k,k'} \Adj^{\ell} \x^{(k)},
\end{equation}
where the response $\Y$ is a  $N \times K'$ matrix, $\y^{(k')}$ is the $k'$th column of $\Y$ and $ \Hh_{\ell,k,k'}$ is trained during learning. We use the symbol $*_w$, instead of $*_v$, to emphasize that the filter coefficients in the proposed graph convolution~\eqref{eq:neural_graph_convolution} are trained in end-to-end learning while the filter coefficients in~\eqref{eq:graph_convolution} are manually designed. This graph convolution~\eqref{eq:neural_graph_convolution} is analogous to conventional convolution~\eqref{eq:neural_convolution}; that is, the output signal coefficient at each vertex is a weighted average of the signal coefficients at neighboring vertices. At the same time,~\eqref{eq:neural_graph_convolution} is a multi-channel extension of~\eqref{eq:graph_convolution}: when $K=K'=1$,~\eqref{eq:neural_graph_convolution} degenerates to~\eqref{eq:graph_convolution}. To make the notation consistent, when $K=K'=1$, we still use
$
\y \ = \ \Hh *_w \x
= \sum_{\ell = 1}^L \Hh_{\ell,1,1} \Adj^{\ell} \x,
$ where $\Hh \in \R^{L \times 1 \times 1}$.

An equivalent representation of~\eqref{eq:neural_graph_convolution} is, 
\begin{equation*}
\Y  \ = \ \sum_{\ell=1}^L \Adj^{\ell} \X \HH^{(\ell)},
\end{equation*}
where $\HH^{(\ell)}$ is a $K \times K'$ trainable matrix with $\HH^{(\ell)}_{k,k'} = \Hh_{\ell,k,k'}$.  Through multiplying with $\HH^{(\ell)}$, we transform $\X$ to a high-dimensional feature space.  After that, we diffuse the new features over the vertex domain according to the graph adjacency matrix $\Adj$. A special case is when $L=1$, in which case~\eqref{eq:neural_graph_convolution} degenerates to
\begin{equation}
\label{eq:neural_graph_convolution_matrix}
\Y  \ = \ \Adj \X \HH.
\end{equation}
This graph convolution is actively used in semi-supervised vertex classification~\cite{KipfW:17}; see theoretical comparisons between~\eqref{eq:neural_graph_convolution_matrix} and~\eqref{eq:neural_graph_convolution} in Section~\ref{sec:graph_convolution_theory}.

\subsection{Weight-sharing mechanism}
Previously, we obtained graph convolution through a mathematical analogy with standard convolution. However, the definitions of the neighborhoods are clearly different in the conventional convolution~\eqref{eq:neural_convolution}
and graph convolution~\eqref{eq:neural_graph_convolution}. For a time-series, each shift order $\ell$ introduces one neighbor for each time stamp. Given a cyclic convolution of length $L$, each time stamp has $L$ distinct neighbors and is associated with $L$ corresponding filter coefficients. On the other hand, for a graph signal, each graph shift order $\ell$ might introduce zero, one or multiple neighbors for each vertex. The number of neighbors depends on the local graph structure. Given a graph convolution of length $L$, those $L$ filter coefficients are insufficient to reflect distinct weights for all the neighbors.  This difference in the neighborhood definition distinguishes~\eqref{eq:neural_convolution} and~\eqref{eq:neural_graph_convolution} since a vertex cannot adjust the contribution from each of its neighbors individually. Here we propose a new graph convolution to fill this gap.

To make the graph convolution more flexible and powerful, we consider updating the edge weights in the given graph adjacency matrix. In this way, each vertex will have different impacts on its neighbors, just like in conventional convolution. A straightforward approach is to introduce a mask matrix $\Psi^{(\ell,k,k')} \in \R^{N \times N}$, expanding a single filter coefficient $\Hh_{\ell,k,k'}$ in~\eqref{eq:neural_graph_convolution} to a matrix of coefficients~\cite{IsufiGR:20}.  We then use $\Psi^{(\ell,k,k')} \odot \Adj^\ell$ to replace $\Hh_{\ell,k,k'} \Adj^\ell$ in~\eqref{eq:neural_graph_convolution}. However, when the graph size, $N$, is large, training $O(N^2)$ parameters in $\Psi^{(\ell,k,k')}$ is computationally difficult.

To reduce the number of training parameters, we  allow all the edges to share the same set of weights, which is similar to the weight-sharing mechanism in conventional convolution~\cite{Goodfellow:2016}. We aim to design a kernel function to parameterize each edge weight. For example, for 2D convolution in image processing, the key is to use a local kernel function to map the relative difference between two pixel coordinates to a weight. For graphs, we assign a coordinate to each vertex and use a local kernel function to map the relative difference between two vertex coordinates to an edge weight.

\begin{figure}[thb]
  \begin{center}
    \begin{tabular}{cc}
   \includegraphics[width=0.3\columnwidth]{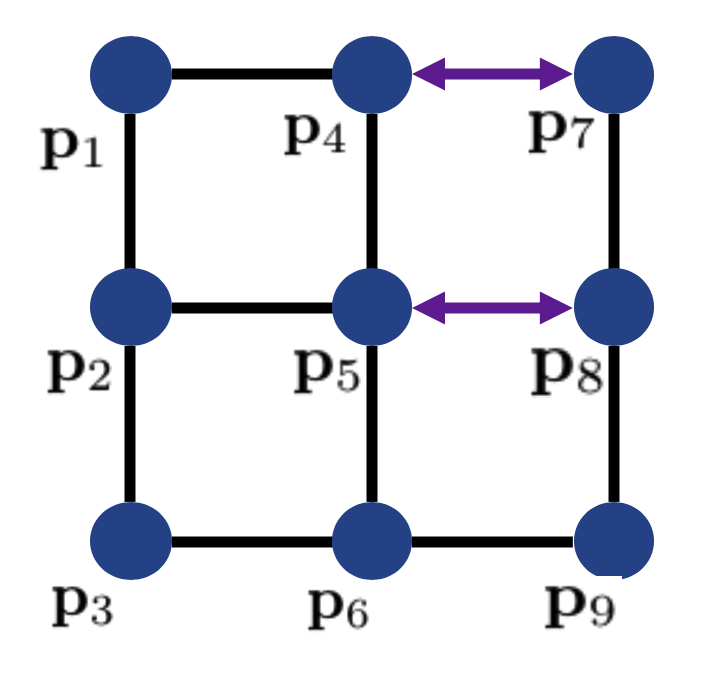} &      \includegraphics[width=0.6\columnwidth]{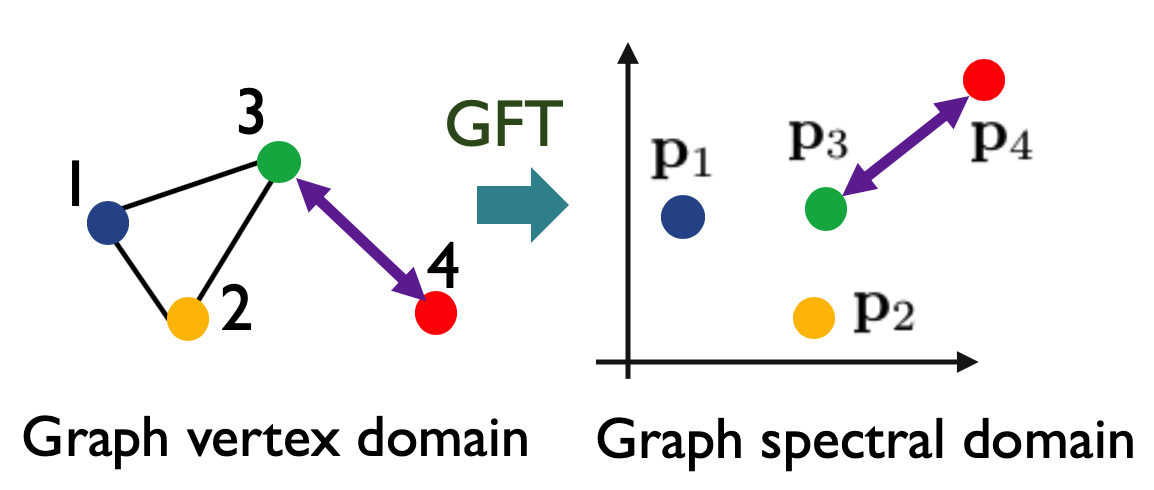}
    \\
    {\small (a) Pixel coordinates.} &  {\small  (b) Vertex coordinates.}
  \end{tabular}
\end{center}
\caption{\label{fig:edge_weight} Vertex coordinates and edge weights in images and graphs. Plot (a) shows a 2D image, where each pixel natually has a pixel coordinate on a 2D lattice. A weight in a 2D convolution is determined by the relative difference between a pair of pixel coordinates. For example, the weight between the $4$th and the $7$th pixels is equal to the weight between the $5$th and the $8$th pixels because the relative pixel coordinates are the same; that is, $\psi_w(\p_4 - \p_7) = \psi_w(\p_5 - \p_8)$; Plot (b) shows a irregular graph, where each vertex can be mapped to a vertex coordinate in the graph spectral domain through the graph Fourier transform (GFT). An edge weight in the proposed graph convolution is determined by the relative difference between a pair of vertex coordinates~\eqref{eq:edge_weight}.}
\end{figure}

The vertex coordinates can be obtained through the graph Fourier transform~\cite{SandryhailaM:13}. Let the eigendecomposition of the graph adjacency matrix  be
$
\Adj \ = \ \Vm \Lambda \Vm^{T},
$
where $\Lambda = {\rm diag}(\lambda_1,\ldots, \lambda_{N})$ is a diagonal matrix of $N$ eigenvalues and $\Vm$ is the matrix of corresponding eigenvectors. The eigenvalues of $\Adj$ represent the~\emph{graph frequencies} and the eigenvectors form the~\emph{graph Fourier basis}. The coordinate of the $i$th vertex is the row vector of the truncated graph Fourier basis, $\p_i = \begin{bmatrix} \Vm_{i,1} & \Vm_{i,2}  &\ldots & \Vm_{i,p}  \end{bmatrix}^T \in \R^p$, where $p \leq N$ is a hyperparameter. Through the graph Fourier transform, we map the information of each vertex from 
the graph vertex domain to the graph spectral domain.

Next, we assume that the edge weight between the $i$th and the $j$th vertices is parameterized by a kernel function:
\begin{equation}
\label{eq:edge_weight}
\Psi_{i,j} = \psi_w \left([\p_j - \p_i] \right) \in \R,
\end{equation}
where $\psi_w(\cdot)$ is a trainable function, which can be implemented by a multilayer perceptron (MLP)~\cite{Goodfellow:2016}. Given a pair of vertex coordinates, we convert their relative coordinate difference to a scalar that reflects the corresponding edge weight. In 2D convolution for images, the convolution kernel function is independent of the absolute pixel coordinates. Specifically, when we set the origin of the kernel function to a pixel, the weight from each of the pixel's neighbors depends on the relative coordinate difference between a neighboring pixel and the origin. Similarly, here we use the relative coordinate difference as the input because it allows the kernel function $\psi_w(\cdot)$ to be irreverent to the exact coordinate and to be applied to arbitrary edges. In other words, the proposed graph convolution is equivalent to the continuous convolution operated in the graph spectral domain. We name this operation the edge-weight-sharing graph convolution because all the edge weights share the same kernel function with the same training parameters; see Fig.~\ref{fig:edge_weight}.

Note that a few previous works also consider learning edge weights. For example, EdgeNet considers each edge weight as an independent trainable parameter~\cite{IsufiGR:20}. However, the number of trainable parameters depends on the graph size, which is computationally expensive. Graph attention networks learn edge weights through the attention mechanism. Each edge weight is parameterized by a kernel function, whose inputs are graph signals~\cite{VelivckovicCCRLB:18}. Here we consider a different parameterization: the input of a kernel function 
is the difference between a pair of vertex coordinates, which relies on the graph structure and is independent of the graph signals. This approach leverages the graph spectral information, which fuses both global and local information on graphs. The number of trainable parameters depends on the kernel function and is independent of the graph size.

We finally propose the ~\emph{edge-weight-sharing graph convolution (EWS-GC)} as $$
\Y = \Hh *_{a} \X
$$ 
with $k'$th output channel
\begin{equation}
\label{eq:weight_sharing_graph_convolution}
\y^{(k')}  \ = \ \sum_{\ell=1}^L \sum_{k=1}^K \left( \Psi^{(\ell,k,k')} \odot \Adj^\ell \right) \x^{(k)} \in \R^N,
\end{equation} 
where the response $\Y$ is a  $N \times K'$ matrix, $\Psi^{(\ell,k,k')} \in \R^{N \times N}$ is an edge-weight matrix whose elements are trainable and follow from~\eqref{eq:edge_weight}. Here the graph filter coefficients form a five-mode tensor $\Hh \in \R^{L \times K \times K' \times N  \times N }$  with $\Hh_{\ell,k,k',i,j} = \Psi^{(\ell,k,k')}_{i,j}$. In the implementation of $\Psi^{(\ell,k,k')}$, we only need to compute those entries  whose corresponding entries are nonzero in $\Adj^{\ell}$. The edge-weight-sharing graph convolution still relies on the given graph structure to propagate information, but it has flexibility to adjust the edge weights. Note that we use the symbol $*_a$, instead of $*_w$, to emphasize that~\eqref{eq:weight_sharing_graph_convolution} works with trainable edge weights while~\eqref{eq:neural_graph_convolution} assumes fixed edge weights.

We can represent~\eqref{eq:weight_sharing_graph_convolution} from another perspective. Let $\E^{(\ell)}$ be the edge set associated with the polynomial of the graph adjacency matrix $\Adj^{\ell}$. For an arbitrary edge $e = (v_i,v_j) \in \E^{(\ell)}$, its indicating matrix $\delta_e$ is defined as a $N \times N$ matrix whose elements are
\begin{equation*}
 (\delta_e)_{i',j'} = 
  \left\{ 
    \begin{array}{rl}
      \Adj^{\ell}_{i',j'}, & \mbox{if}~e = (v_{i'},v_{j'});\\
      0, & \mbox{otherwise}.
  \end{array} \right.
\end{equation*}
The subscript $e$ indicates an edge and $\delta_e$ is a one-hot matrix, only activating the element specified by the edge $e$. This edge is associated with a $K \times K'$ trainable matrix $\HH^{(e)}$ with elements $\HH^{(e)}_{k,k'} = \Hh_{\ell,k,k',i,j}$. The equivalent representation of~\eqref{eq:weight_sharing_graph_convolution} is\footnote{Equation~\eqref{eq:weight_sharing_graph_convolution_mtx}  suggests a randomized implementation 
of the edge-weight-sharing graph convolution. Instead of using the entire edge sets, we can randomly sample a subset of edges and approximate the exact value of~\eqref{eq:weight_sharing_graph_convolution_mtx} as $\Y  \ \approx \  \sum_{e \in \M} \delta_{e} \X \HH^{(e)}$. The edge set $\M \subset  \E^{(1)} \cup \E^{(2)} \cup \cdots \cup \E^{(L)}$ is obtained through edge sampling, which can be implemented via random walks~\cite{PerozziAS:14}.}
\begin{equation}
\label{eq:weight_sharing_graph_convolution_mtx}
\Y  \ = \ \sum_{\ell=1}^L \sum_{e \in \E^{(\ell)}} \delta_{e} \X \HH^{(e)},
\end{equation}
where $\Adj^{\ell} = \sum_{e \in \E^{(\ell)}} \delta_e$. The first summation considers all edge sets and the second summation considers all edges. We parse the entire graph to a collection of edges and the effect of each edge is reflected through the corresponding trainable matrix $\HH^{(e)}$.

The edge-weight-sharing graph convolution~\eqref{eq:weight_sharing_graph_convolution} is a specific type of edge-variant graph convolution, which leverages the weight-sharing mechanism to significantly reduce the number of trainable parameters. Indeed, \eqref{eq:neural_graph_convolution} is a special case of~\eqref{eq:weight_sharing_graph_convolution} when all elements in each $\Psi^{(\ell,k,k')}$ have the same value. Similarly,~\eqref{eq:neural_graph_convolution_matrix} is a special case of~\eqref{eq:neural_graph_convolution} when the filter length $L=1$.

We can apply the proposed edge-weight-sharing graph convolution~\eqref{eq:weight_sharing_graph_convolution} as a substitute to conventional graph filtering in the denoising problem~\eqref{eq:graph_signal_denoising_optimization}. In the next section, we use~\eqref{eq:weight_sharing_graph_convolution} as a building block for graph neural networks. Its associated trainable parameters will be updated in an end-to-end learning process.

\subsection{Analysis of graph convolution}
\label{sec:graph_convolution_theory}
Here we provide theoretical analysis of graph convolutions. We first show a drawback of simple graph convolution~\eqref{eq:neural_graph_convolution_matrix} and then discuss the benefit brought by the graph convolution~\eqref{eq:neural_graph_convolution} and the edge-weight-sharing graph convolution~\eqref{eq:weight_sharing_graph_convolution}.

\begin{myThm}
\label{thm:graph_convolution_1}
Let $\X$ be a matrix of $K$ graph signals. Let $\HH *_w \X = \Adj \X \HH$~\eqref{eq:neural_graph_convolution_matrix} be a graph convolution, where $\HH$ is a $N \times K$ trainable matrix. Let $\Y = \lim_{t \rightarrow +\infty} \HH^{[t]} *_w \cdots \HH^{[2]} *_w \HH^{[1]} *_w \X$ be the matrix of $K$ output graph signals after applying the graph convolution infinitely many times. Suppose $|\lambda_{\rm max}(\Adj)| \neq |\lambda_{\rm min}(\Adj)|$. Then, $$\text{rank}(\Y)=1,$$ independent of $\HH^{[t]}$.
\end{myThm}
\begin{proof} 
We can rewrite the output matrix as
\begin{eqnarray*}
\Y & = & \lim_{t \rightarrow +\infty} \HH^{[t]} *_w \cdots \HH^{[2]} *_w \HH^{[1]} *_w \X
\\
& = & \lim_{t \rightarrow +\infty}
\Adj^{t} \X \HH^{[1]} \HH^{[2]} \cdots \HH^{[t]}
\\
& = & \Vm \lim_{t \rightarrow +\infty}
 \Lambda^t \Vm^T \X \HH^{[1]} \HH^{[2]} \cdots \HH^{[t]},
\end{eqnarray*}
where $\Vm$ and $\Lambda$ are the eigenvector matrix and eigenvalue matrix of the graph shift, respectively. Let the eigenvalues along the diagonal of $\Lambda$ be descendingly ordered. Because of the normalization of the graph shift, the eigenvalues of $\Adj$ satisfy $1 \geq  \lambda_1 \geq \lambda_2 \geq \cdots \geq \lambda_N \geq -1$. Since $|\lambda_{\rm max}(\Adj)| \neq |\lambda_{\rm min}(\Adj)|$, we have either $\lambda_1 = 1$ or $\lambda_N = -1$. Without losing generality, we assume $\lambda_1 = 1$ and $\lambda_N > -1$. Denote $\Z^{[t]} = \Vm^T \X \HH^{[1]} \HH^{[2]} \cdots \HH^{[t]}$ for any given $t$. Then, since $\lambda_1 = 1$ and $\lambda_i^t \rightarrow 0$ for $i \neq 1$, we have
\begin{equation*}
  \Y \ = \  \vv_1
\begin{bmatrix}
\Z^{[t]}_{11} &  \Z^{[t]}_{12} & \cdots &  \Z^{[t]}_{1N}
\end{bmatrix}. 
\end{equation*}
This concludes that the rank of $\Y$ is equal to one. 
\end{proof}
Theorem~\ref{thm:graph_convolution_1} shows that the graph convolution~\eqref{eq:neural_graph_convolution_matrix} can lead to a trivial output even with a huge amount of training parameters. The output graph signal after infinite-time graph convolution will always be proportional to the eigenvector associated with either the largest or the smallest eigenvalue. This indicates the limited power of~\eqref{eq:neural_graph_convolution_matrix}. We next show that the proposed graph convolution~\eqref{eq:neural_graph_convolution} does not suffer from this issue.

\begin{myThm}
\label{thm:graph_convolution_2}
Let $\X$ be a matrix of $K$ graph signals. Let $\Hh *_w \X = \sum_{\ell=1}^L \Adj^{\ell} \X \HH^{(\ell)}$~\eqref{eq:neural_graph_convolution_matrix} be a graph convolution, where $\HH^{(\ell)}$ is a $K \times K$ trainable matrix. Let $\Y = \lim_{t \rightarrow +\infty} \Hh^{[t]} *_w \cdots \Hh^{[2]} *_w \Hh^{[1]} *_w \X$ be the matrix of $K$ output graph signals after graph convolution infinite times. Suppose that (i) the filter length $L \geq Q$, where $Q$ is the number of distinct eigenvalues of the graph shift $\Adj$, and (ii) none of the eigenvalues of $\Adj$ is equal to zero.
Then, $$\text{rank}(\Y) = \text{rank}(\X),$$ with a careful design of $\Hh^{[t]}$.
\end{myThm}
\begin{proof}
We consider a constructive proof; that is, we design a specific $\Hh^{[t]}$ to make $\text{rank}(\Y) = \text{rank}(\X)$. Without losing generality, in $\Hh^{[1]} *_w \X$, we set the trainable matrix $\HH^{(\ell)} = h_{\ell} \HH$, where $\HH$ is a $K \times K$ full-rank matrix and $h_{\ell}$ is a scalar variable. Then, 
\begin{eqnarray*}
&& \sum_{\ell=1}^L \Adj^{\ell} \X \HH^{(\ell)} \ = \
\sum_{\ell=1}^L h_{\ell} \Adj^{\ell} \X \HH
\\
& = &
\Vm \begin{bmatrix}
\sum_{\ell=1}^L h_{\ell} \lambda_1^{\ell} & 0 & \cdots & 0 \\
0 & \sum_{\ell=1}^L h_{\ell} \lambda_2^{\ell} & \cdots & 0 \\
\vdots & \vdots & \ddots & \vdots \\ 
0 & 0 & \cdots & \sum_{\ell=1}^L h_{\ell} \lambda_N^{\ell}
\end{bmatrix}  \Vm^T \X \HH.
\end{eqnarray*}
We aim to design $h_{\ell}$ to satisfy $\sum_{\ell=1}^L h_{\ell} \lambda_i^{\ell} = 1$ for all $i$. Without losing generality, let $\lambda_1, \lambda_2, \cdots, \lambda_Q$ be $Q$ distinct and nonzero eigenvalues of $\Adj$. We then need to solve
\begin{equation*}
\begin{bmatrix}
\lambda_1 & 0 & \cdots & 0 \\
0 & \lambda_2 & \cdots & 0 \\
\vdots & \vdots & \ddots & \vdots \\ 
0 & 0 & \cdots & \lambda_Q \\
\end{bmatrix}  
\begin{bmatrix}
1 & \lambda_1 & \cdots & \lambda_1^{L-1} \\
1 & \lambda_2 & \cdots & \lambda_2^{L-1} \\
\vdots & \vdots & \ddots & \vdots \\ 
1 & \lambda_N & \cdots & \lambda_Q^{L-1} \\
\end{bmatrix} 
\begin{bmatrix}
h_1 \\ h_2  \\ \vdots \\ h_L
\end{bmatrix} = 
\begin{bmatrix}
1 \\
1  \\
\vdots \\ 
1
\end{bmatrix} .
\end{equation*}
The diagonal matrix has full rank as all the elements are  nonzero; the Vandermonde matrix also has full row rank as $L \geq Q$ and all the $\lambda_i$'s are distinct. We thus can  design $h_{\ell}$ to achieve equality and then $\sum_{\ell=1}^L \Adj^{\ell} \X \HH^{(\ell)} = \X \HH$. Since $\HH$ is a full-rank square matrix, $\text{rank}(\Hh^{[1]} *_w \X) = \text{rank}(\sum_{\ell=1}^L \Adj^{\ell} \X \HH^{(\ell)}) =  \text{rank}(\X \HH) = \text{rank}(\X)$; in other words, it is possible to design $\Hh^{[1]}$ to allow the output after each graph convolution to have the same rank as the input. We can apply the same technique arbitrarily many times to achieve $\text{rank}(\Y) = \text{rank}(\X)$.
\end{proof}
We see that the graph convolution~\eqref{eq:neural_graph_convolution} is much more powerful than its special case~\eqref{eq:neural_graph_convolution_matrix}. However,~\eqref{eq:neural_graph_convolution} also has its own limitation.

\begin{myThm}
\label{thm:graph_convolution_3}
Let $\X = {\bf 1}_N \x^T$ be a matrix of $K$ constant graph signals, where ${\bf 1}_N$ is a $N$-dimensional all-one vector and $\x$ is a $K$-dimensional vector. Let $\Hh *_w \X = \sum_{\ell=1}^L \Adj^{\ell} \X \HH^{(\ell)}$ be the graph convolution, where $\HH^{(\ell)}$ is a $N \times K$ matrix. Let $\Adj = \Id$ be a self-loop graph. Then, the output after graph convolution is always a constant graph signal.
\end{myThm}
\begin{proof}
Since the underlying graph is a self-loop graph, we have
\begin{equation*}
\Hh *_w \X \ = \ \sum_{\ell=1}^L \Id \X \HH^{(\ell)}
 =  {\bf 1}_N \left(\x^T \sum_{\ell=1}^L \HH^{(\ell)} \right).
\end{equation*}
We cannot change the constant graph signals through graph convolution~\eqref{eq:neural_graph_convolution}.
\end{proof}
We see that the graph convolution~\eqref{eq:neural_graph_convolution} is weak as it can only output constant graph signals when the input are constant graph signals. It is straightforward to show that the proposed edge-weight-sharing  graph convolution~\eqref{eq:weight_sharing_graph_convolution_mtx} can fix this issue by adjusting each edge weight. For a self-loop graph $\Adj = \Id$,~\eqref{eq:weight_sharing_graph_convolution_mtx} is then
\begin{eqnarray*}
\sum_{\ell=1}^L \sum_{e \in \E^{(\ell)}} \delta_{e} \X \HH^{(e)} 
= {\rm diag}(\ee) {\bf 1}_N \left(\x^T \sum_{\ell} \HH^{(\ell)}\right)
= \ee \left(\x^T \sum_{\ell} \HH^{(\ell)}\right),
\end{eqnarray*}
where $\ee$ is a $N$-dimensional vector, which can be trained to reflect the edge weights. Therefore, we can change constant graph signals through~\eqref{eq:weight_sharing_graph_convolution_mtx}.

Finally, we consider the permutation-equivariant property of the edge-weight-sharing graph convolution~\eqref{eq:weight_sharing_graph_convolution_mtx}. Let $\J \in \R^{N \times N}$ be a permutation matrix. After permutation, a graph adjacency matrix $\Adj \in \R^{N \times N}$ and a $K$-channel graph signal $\X \in \R^{N \times K}$ become 
$\J \Adj \J^T$ and $\J \X$, respectively.
\begin{myThm}
The edge-weight-sharing graph convolution~\eqref{eq:weight_sharing_graph_convolution_mtx} is permutation equivariant.  Suppose that a kernel function $\psi_w(\cdot)$ is fixed and given. Then,
$$\J \left( \Hh *_{a} \X \right) = \Hh *_{a} \left(\J\X\right).$$
\end{myThm}
\begin{proof}
We first show the effect of permutation on the trainable edge-weight matrix~\eqref{eq:edge_weight}. When we permute the graph structure, the vertex coordinates permute accordingly. Therefore, after permutation, a edge-weight matrix $\Psi^{(\ell,k,k')}$~\eqref{eq:edge_weight} becomes $\J \Psi^{(\ell,k,k')} \J^T$. The $k$th channel of $\Hh *_{a} \left(\J\X \right)$ becomes
\begin{eqnarray*}
&& \Big( \Hh *_{a}\left(\J\X \right) \Big)^{(k)} 
\\
& = & \sum_{\ell=1}^L \sum_{k=1}^K \left( \left(  \J \Psi^{(\ell,k,k')} \J^T \right) \odot \left(\J\Adj\J^T \right)^\ell \right) \J \x^{(k)}
\\
& = &  \J \sum_{\ell=1}^L \sum_{k=1}^K \left( \Psi^{(\ell,k,k')}  \odot \Adj^\ell \right) \x^{(k)}
\\
& = & \Big( \J  \left( \Hh *_{a} \X \right) \Big)^{(k)},
\end{eqnarray*}
which is the the $k$th channel of $\J \left( \Hh *_{a} \X \right)$.
\end{proof}
The permutation-equivariant property is important because it ensures that reordering of the vertices will not effect the training of the edge-weight-sharing graph convolution. With the new graph convolution operation, we now move to propose an interpretable framework for designing graph neural networks.

\section{Graph unrolling networks}
\label{sec:GUN}
In Section~\ref{sec:formulation}, we mathematically formulated the task of graph signal denoising~\eqref{eq:graph_signal_denoising_optimization}. In this section, we aim to solve~\eqref{eq:graph_signal_denoising_optimization} through algorithm unrolling and propose a general framework for developing graph unrolling networks. The core strategy is to follow iterative algorithms and then use the trainable edge-weight-sharing graph  convolution proposed in Section~\ref{sec:graph_convolution} to  substitute fixed graph filtering. We further consider two specific architectures of graph unrolling networks by unrolling graph sparse coding and graph trend filtering.

\subsection{General framework}
%To handle graph signal denoising, a common limitation of many representation-based or optimization-based methods is that an appropriate prior is usually unknown for a specific set of graph signals. Here we suggest learning a data-driven prior from the noisy measurements through graph neural networks. Instead of using an arbitrary neural network architecture without any design principle, we consider using algorithm unrolling to derive an interpretable graph neural network; that is, we unroll an iterative algorithm for solving~\eqref{eq:graph_signal_denoising_optimization} into a graph neural network by mapping each iteration into a single network layer and stacking multiple layers together. In this manner, this graph neural network can be naturally interpreted as a parameter optimized algorithm, and the training process is equivalently the denoising process.

Consider a general iterative algorithm to solve the graph signal denoising problem~\eqref{eq:graph_signal_denoising_optimization} based on the half-quadratic splitting algorithm. The basic idea is to perform variable-splitting and then alternating minimization on the penalty function~\cite{WangYYZ:08,LiTME:19}.

Introduce two auxiliary variables $\y = \Pj \x$ and $\z = \Q \s$ and then reformulate~\eqref{eq:graph_signal_denoising_optimization} as
\begin{eqnarray*}
\label{eq:graph_signal_denoising_optimization_reform}
  && \min_{\s \in \R^N} \frac{1}{2} \left\|\t - \x \right\|_2^2 +  u( \y) + r( \z),
	\\ \nonumber
	&& {\rm subject~to~~} \x \ =  \  \h *_v \s,~~\y = \Pj \x,~~\z = \Q \s.
\end{eqnarray*}
The penalty function is
\begin{eqnarray*}
 L(\x, \s, \y, \z) & = &  \frac{1}{2} \left\|\t - \x \right\|_2^2 +  u( \y) + r( \z) +  \frac{\mu_1}{2} \left\| \x - \h *_v \s \right\|_2^2 \\
 &&   
 + \frac{\mu_2}{2} \left\| \y - \Pj \x \right\|_2^2 +  \frac{\mu_3}{2} \left\| \z - \Q \s \right\|_2^2,
\end{eqnarray*}
where $\mu_1, \mu_2, \mu_3$ are appropriate step sizes. We then alternately minimize $L(\x, \s, \y, \z)$ over $\x, \s, \y, \z$, leading to the following updates:
\begin{subequations}
\label{eq:general_admm}
\begin{eqnarray}
\label{eq:general_admm_x}
    \x & \leftarrow &  \widetilde{\Pj} \bigg(  \mu_1 \h *_v \s + \t + \mu_2 \Pj^T \y \bigg),
    \\
\label{eq:general_admm_s}
    \s & \leftarrow &  \widetilde{\Q} \bigg( \mu_1 \h *_v^T \x + \mu_3 \Q^T \z \bigg),
    \\
\label{eq:general_admm_y}
    \y & \leftarrow &  \arg \min_{\y} \frac{\mu_2}{2} \left\| \y - \Pj \x \right\|_2^2 +  u( \y),
    \\
\label{eq:general_admm_z}
    \z & \leftarrow &  \arg \min_{\z} \frac{\mu_3}{2} \left\| \z - \Q \s  \right\|_2^2 +  r( \z),
\end{eqnarray}
\end{subequations}
where $$\widetilde{\Pj} = (\Id+\mu_1\Id + \mu_2 \Pj^T \Pj )^{-1}$$ and 
$$\widetilde{\Q} = ( \mu_1 \sum_{\ell'=1}^L h_{\ell'} \Adj^{\ell'} \sum_{\ell=1}^L h_{\ell} \Adj^{\ell}  + \mu_3 \Q^T \Q )^{-1}.$$ 
Intuitively,~\eqref{eq:general_admm_x} denoises a graph signal by merging information from the original measurements $\t$, filtered graph signals $\h *_v \s$ and the auxiliary variable $\y$;~\eqref{eq:general_admm_s} generates a base graph signal through graph deconvolution; and~\eqref{eq:general_admm_y} and~\eqref{eq:general_admm_z} solve two proximal functions with regularization $u(\cdot)$ and $r(\cdot)$, respectively.

To unroll the iteration steps~\eqref{eq:general_admm}, we consider two major substitutions. First, we replace the fixed graph convolution by the trainable edge-weight-sharing  graph convolution~\eqref{eq:weight_sharing_graph_convolution}. Second, we replace the sub-optimization problems in~\eqref{eq:general_admm_y} and~\eqref{eq:general_admm_z} by a trainable neural network. 

A unrolling layer for denoising a graph signal is then,
\begin{subequations}
\label{eq:GUN_single}
\begin{eqnarray}
\label{eq:GUN_x}
\x & \leftarrow & 
\Aa *_a \s + \Bb *_a \t + \Cc *_a \left(\Pj^T \y\right),
\\
\label{eq:GUN_s}
\s & \leftarrow & \Dd *_a \x + \Ee *_a \left(\Q^T \z \right),
\\
\label{eq:GUN_y}
\y & \leftarrow & {\rm NN}_{u} \left(\Pj \x \right),
\\
\label{eq:GUN_z}
\z & \leftarrow & {\rm NN}_{r} \left(\Q \s \right),
\end{eqnarray}
\end{subequations}
where $\Aa *_a$, $\Bb *_a$, $\Cc *_a$, $\Dd *_a$, and $\Ee *_a$ are individual edge-weight-sharing  graph convolutions with filter coefficients that are trainable parameters, and ${\rm NN}_u(\cdot)$ and ${\rm NN}_r(\cdot)$ are two neural networks, which involve trainable parameters. Intuitively,~\eqref{eq:GUN_x} and~\eqref{eq:GUN_s} are neural-network implementations of~\eqref{eq:general_admm_x} and~\eqref{eq:general_admm_s}, respectively,  replacing fixed graph convolutions $\h *_v$ by  trainable edge-weight-sharing  graph convolutions~\eqref{eq:weight_sharing_graph_convolution}; and~\eqref{eq:GUN_y} and~\eqref{eq:GUN_z} are neural-network implementations of the proximal functions~\eqref{eq:general_admm_y} and~\eqref{eq:general_admm_z}, respectively, using neural networks to solve sub-optimization problems; see similar substitutions in~\cite{GregorL:10,SolomonCZYHLSE:18,LiTME:19,LiTGME:20}.

Instead of following the exact mathematical relationship in~\eqref{eq:general_admm}, we allow trainable operators to adaptively learn from data, usually reducing a lot of computation. The implementations of~\eqref{eq:GUN_y} and~\eqref{eq:GUN_z} depend on specific regularization terms, $u(\cdot)$ and $r(\cdot)$. For some $u(\cdot), r(\cdot)$, we might end up with an analytical form for~\eqref{eq:GUN_y} and~\eqref{eq:GUN_z}. We will show two examples in Sections~\ref{sec:graph_sparse_coding} and~\ref{sec:graph_trend_filtering}.

One characteristic of neural networks is to allow feature learning in a high-dimensional space. Instead of sticking to a single channel, we can easily extend~\eqref{eq:GUN_single} to handle multiple input noisy graph signals and enable multiple-channel feature learning. The corresponding $b$th unrolling layer of multi-channel graph signals is
%\begin{subequations}
\begin{eqnarray}
\label{eq:GUN_multiple_x}
\X^{(b)} & \leftarrow & 
\Aa *_a \Ss^{(b-1)}+\Bb *_a \T+\Cc *_a \left(\Pj^T \Y^{(b-1)} \right),
\nonumber \\
\label{eq:GUN_multiple_s}
\Ss^{(b)} & \leftarrow & \Dd *_a \X^{(b)}+\Ee *_a \left(\Q^T \Z^{(b-1)} \right),
\nonumber \\
\label{eq:GUN_multiple_y}
\Y^{(b)} & \leftarrow & {\rm NN}_{u} \left(\Pj \X^{(b)} \right),
\nonumber \\
\label{eq:GUN_multiple}
\Z^{(b)} & \leftarrow & {\rm NN}_{r} \left(\Q \Ss^{(b)} \right),
\end{eqnarray}
%\end{subequations}
where $\T \in \R^{N \times K}$ is a matrix of $K$ noisy graph signals,  $\Z^{(b)} \in \R^{N \times d^{(b)}}$ is the intermediate feature matrix, and $\Ss^{(b)} \in \R^{N \times D^{(b)}}$ is the output matrix of the $b$th computational block. The feature dimensions $d^{(b)}, D^{(b)}$ are hyperparameters of the network; see a graph unrolling layer in Fig.~\ref{fig:generic_unrolling}.

\begin{figure}[thb]
  \begin{center}
   \includegraphics[width=1.00\columnwidth]{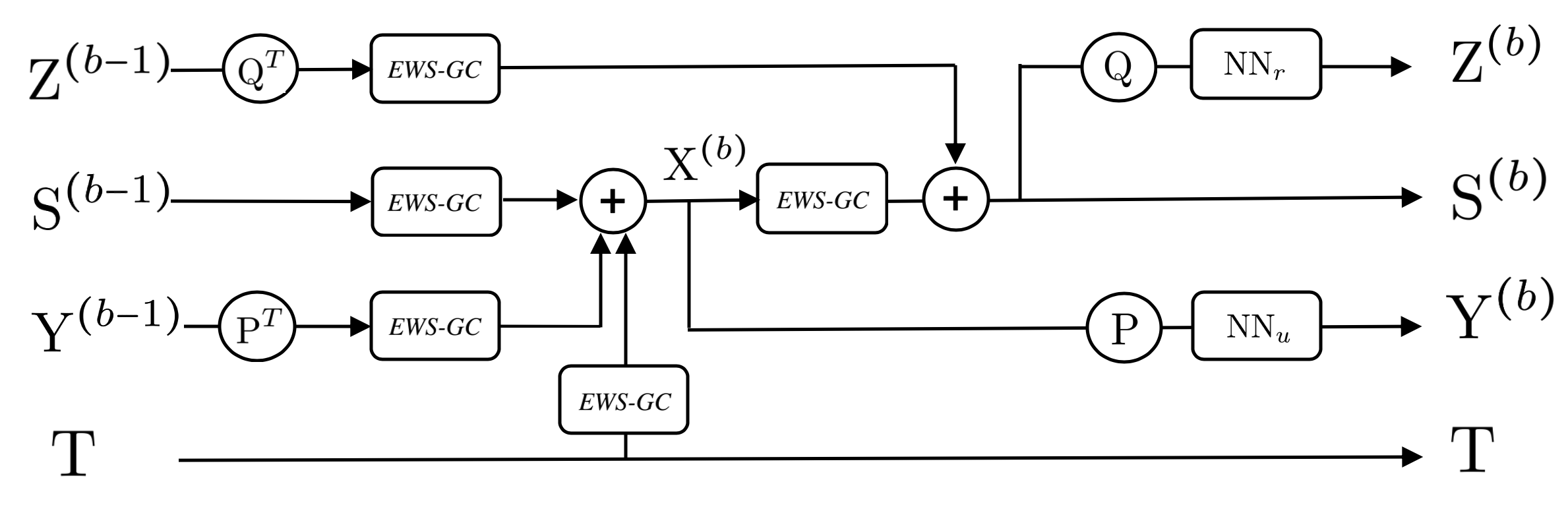}
\end{center}
\caption{\label{fig:generic_unrolling} A generic graph unrolling layer for graph signal denoising~\eqref{eq:GUN_multiple}, which is one computational block in a graph unrolling network. Given the raw measurement $\T$, the proposed unrolling layer updates $\Z$, $\Ss$, and $\Y$.}
\end{figure}

To build a complete network architecture, we initialize $\Z^{(0)}, \Ss^{(0)}, \Y^{(0)}$ to be all-zero matrices and sequentially stack $B$ unrolling layers~\eqref{eq:GUN_multiple}. This is hypothetically equivalent to running the iteration steps~\eqref{eq:general_admm} for $B$ times. Through optimizing trainable parameters in edge-weight-sharing  graph convolutions and two sub-neural-networks, we obtain the denoised output $\widehat{\X} = \X^{(B)}$.

Here the trainable parameters come from two parts, including filter coefficients in each edge-weight-sharing  graph convolution and the parameters in the neural networks~\eqref{eq:GUN_y} and~\eqref{eq:GUN_z}. Through optimizing those parameters, we can capture complicated priors in the original graph signals in a data-driven manner. To train those parameters, we consider the loss function
\begin{equation}
\label{eq:loss}
   {\rm loss} \ = \  \left\| f\left({\T}\right) - \T \right\|_{F}^2 = \left\|\widehat{\X} - \T \right\|_{F}^2,
\end{equation}
where $\left\|\cdot\right\|_{F}$ is the Frobenius norm, $\widehat{\X}$ is the output of the proposed network $f(\cdot)$, and $\T$ are the original measurements. We then use the stochastic gradient descent to minimize the loss and optimize this network~\cite{Goodfellow:2016}. Note that the noisy measurement $\T$ is used as both input and supervision of the network.

Hypothetically, the loss could be zero when a neural network is trained to be an identity mapping. In other words, the denoised output is the same as the input noisy measurements. In practice, however, this does not happen since the building block of a graph unrolling network is an edge-weight-sharing  graph convolution, whose effect heavily depends on irregular graph structures. This convolution injects implicit graph regularization to the network architecture and the overall optimization problem.

Our algorithm unrolling here is rooted in the  half-quadratic splitting  algorithm. In practice, our optimization problem can be solved using various alternative iterative algorithms, which may lead to distinct network architectures.~\emph{No matter what iterative algorithm is used, the core strategy is to follow the iterative steps and use trainable edge-weight-sharing  graph convolution to substitute fixed, yet computationally expensive graph filtering.} We call a network architecture that follows this strategy a~\emph{graph unrolling network (GUN)}.

Compared to conventional graph signal denoising  algorithms~\cite{ChenSMK:14a, ShumanNFOV:13},  the  proposed  GUN  is able  to  learn  a  variety  of complicated signal priors from given graph signals by leveraging the learning ability of deep neural networks. Compared  to many generic graph neural networks~\cite{KipfW:17}, the  proposed GUN is interpretable by following analytical iterative steps. We unroll an iterative algorithm for solving~\eqref{eq:graph_signal_denoising_optimization} into a graph neural network by mapping each iteration into a single network layer and stacking multiple layers together. In this manner, the proposed GUN can be naturally interpreted as a parameter optimized algorithm.

In the following, we present two special cases of the GUN, which are obtained through unrolling graph sparse coding and graph trend filtering, respectively.  Graph sparse coding is a typical graph-dictionary-based denoising algorithm, where we first design a graph dictionary based on a series of graph filters and then select a few elementary graph signals from a graph dictionary to approximate a noisy graph signal~\cite{HammondVG:11,ShumanWHV:15,ShafipourKM:19}. The unrolling version of graph sparse coding essentially 
combines these two steps in an end-to-end learning process and uses an edge-weight-sharing graph convolution to substitute the predesigned graph dictionary. On the other hand, graph trend filtering is a typical graph-regularization-based denoising algorithm, where we first formulate an optimization problem with graph total variation as an explicit graph regularization and then solve this optimization problem to denoise graph signals~\cite{WangSST:16,VarmaLKC:20}. The unrolling version of graph trend filtering uses a trainable edge-weight-sharing graph convolution to provide an implicit graph regularization and uses end-to-end learning to optimize the trainable parameters. Comparing these two methods, graph trend filtering is designed for piecewise-constant and piecewise-smooth graph signals, so that GUTF is more regularized. On the other hand, graph sparse coding works for a broader class of graph signals, resulting in GUSC being more general, but typically requires more training data. In the experiments, we will see that GUTF achieves better denoising performances than GUSC for simulated data, including smooth graph signals, piecewise-constant graph signals, and piecewise-smooth graph signals. When the number of graph signals increases, the gap between GUTF and GUSC decreases. On the other hand, GUSC achieves better denoising performance than GUTF for real-world data, which has more complicated structure than simulated data.

\begin{algorithm}[h]
  \footnotesize
  \caption{\label{alg:GUSC} Graph unrolling sparse coding (GUSC)}
  \begin{tabular}{@{}ll@{}}
    \addlinespace[1mm]
   {\bf Input} 
      & $\T$  ~~~ matrix of measurements \\
      & $\Adj$  ~~~ graph adjacency matrix \\
      & $B$  ~~~ number of network layers \\
      & $\alpha$ ~~~ hyperparameter \\
     {\bf Output}  
      & $\Xw$  ~~~ matrix of denoised graph signals\\
    \addlinespace[2mm]
    {\bf Function} 
    & {\bf GSC($\T, \Adj, B, \alpha$) } \\
	& $\Ss^{(0)} \leftarrow \bf{0}$    \\	
	& $\{\p_i\}_{i \in \V} \leftarrow $ eigendecomposition of $\Adj$  \\
	& for $b = 1:B$ \\  
	& ~~~$\Aa, \Bb, \Dd, \Ee_{\ell,k,k',i,j} \leftarrow {\rm  MLP}([\p_j - \p_i])$  \\		
    & ~~~$\X^{(b)}  \leftarrow 
\Aa *_a \Ss^{(b-1)}+\Bb *_a \T$,\\
    & ~~~$\Ss^{(b)} \leftarrow  \Dd *_a \X^{(b)}+\Ee *_a \Z^{(b-1)}$, \\
    & ~~~$\Z^{(b)} \leftarrow  S_{\alpha} \left( \Ss^{(b)} \right)$ \\
    & end \\
	& $\Hh_{\ell,k,k',i,j} \leftarrow {\rm  MLP}([\p_j - \p_i])$  \\	
    & $\widehat{\X}  \leftarrow  \Hh *_a \Ss^{(B)}$ \\
    & minimize $\left\| \widehat{\X} - \T \right\|_F^2$ and update all the parameters \\
    & {\bf return} $\Xw$ \\
     \addlinespace[1mm]
  \end{tabular}
\end{algorithm}

\subsection{Graph sparse coding}
\label{sec:graph_sparse_coding}
As discussed in Section~\ref{sec:GSC}, graph sparse coding~\eqref{eq:graph_sparse_coding} considers reconstructing noiseless graph signals through graph filtering and regularizing base graph signals to be sparse. In this setting, $\x \ = \ \h *_v \s \ = \ \sum_{\ell=1}^L h_{\ell} \Adj^\ell \s$, $u( \cdot) = 0, r(\cdot) = \alpha\left\| \cdot \right\|_1$ and $\Pj = \Q = \Id$.  

We can plug in those specifications to~\eqref{eq:GUN_single} and obtain a customized graph unrolling network. We consider three modifications for the customization. First,
we remove the terms related to $\y$ because $u(\y) = 0$ and $\y$ should not effect optimization anymore. Second, 
we remove the terms related to $\x$ because the goal of graph sparse coding is to look for a code $\s$ and there is no need to update an intermediate variable $\x$. Third, we replace~\eqref{eq:GUN_z} by a soft-thresholding function because it is the analytical solution of~\eqref{eq:general_admm_z}~\cite{BoydPCPE:11}. We finally obtain the $b$th unrolling layer customized for graph sparse coding to be
\begin{subequations}
\label{eq:GUSC}
\begin{eqnarray*}
\label{eq:GUSC_x}
\X^{(b)} & \leftarrow & 
\Aa *_a \Ss^{(b-1)}+\Bb *_a \T,
\\
\label{eq:GUSC_s}
\Ss^{(b)} & \leftarrow & \Dd *_a \X^{(b)}+\Ee *_a \Z^{(b-1)},
\\
\label{eq:GUSC_z}
\Z^{(b)} & \leftarrow & S_{\alpha} \left( \Ss^{(b)} \right),
\end{eqnarray*}
\end{subequations}
where $\alpha$ is a hyperparameter and $S_{\cdot}(\cdot)$ is a soft-thresholding function,
\begin{eqnarray*}
\big[ S_{\alpha}(\x) \big]_i \ = \ 
\begin{cases}
x_i - \alpha,~~\text{if}~ x_i - \alpha, 
\\
0, ~~\text{if}~ -\alpha \leq x_i \leq \alpha, 
\\
x_i + \alpha,~~\text{if}~ x_i < -\alpha.
\end{cases}
\end{eqnarray*}
All the training parameters are involved in the edge-weight-sharing  graph convolutions, $\Aa *_a, \Bb *_a, \Cc *_a, \Dd *_a $ and $\Ee *_a$. Since the architecture inherits from both the graph unrolling network and graph sparse coding, we call this architecture a~\emph{graph unrolling sparse coding (GUSC)}. The training paradigm follows the general graph unrolling network; see its overall implemetation in Algorithm~\ref{alg:GUSC}. The hyperparameter $\alpha$ in the soft-thresholding function could be trainable. In the experiments, we find that the performance of a fixed $\alpha$ is slightly better than a trainable $\alpha$; see Section~\ref{sec:simulation}.

\begin{algorithm}[h]
  \footnotesize
  \caption{\label{alg:GUTF} Graph unrolling trend filtering (GUTF)}
  \begin{tabular}{@{}ll@{}}
    \addlinespace[1mm]
   {\bf Input} 
      & $\T$  ~~~ matrix of measurements \\
      & $\Adj$  ~~~ graph adjacency matrix \\
      & $B$  ~~~ number of network layers \\
      & $\alpha$  ~~~ hyperparameter \\
     {\bf Output}  
      & $\Xw$  ~~~ matrix of denoised graph signals\\
    \addlinespace[2mm]
    {\bf Function} 
    & {\bf GTF($\T, \Adj, B, \alpha$) } \\
	& $\X^{(0)} \leftarrow \bf{0}$    \\
	& Obtain $\Delta$ from $\Adj$ via~\eqref{eq:Delta}    \\
	& $\{\p_i\}_{i \in \V} \leftarrow $ eigendecomposition of $\Adj$ \\
	& for $b = 1:B$ \\   
	% & ~~~initialize $\W^{(b)}$   \\
	& ~~~$\Bb_{\ell,k,k',i,j}, \Cc_{\ell,k,k',i,j} \leftarrow {\rm  MLP}([\p_j - \p_i])$  \\
	& ~~~$\Y^{(b)}  \leftarrow  S_{\alpha} \left(  \Delta \X^{(b-1)} \right)$ \\		  
    & ~~~$\X^{(b)} \leftarrow  \Bb *_a \T + \Cc *_a  \left( \Delta^T  \Y^{(b)} \right),$ \\
    & end \\
    & $\widehat{\X}  \leftarrow \X^{(B)}$ \\
    & minimize $\left\| \widehat{\X} - \T \right\|_F^2$ and update all the weights \\
    & {\bf return} $\Xw$ \\
     \addlinespace[1mm]
  \end{tabular}
\end{algorithm}

\subsection{Graph trend filtering}
\label{sec:graph_trend_filtering}
As discussed in Section~\ref{sec:GTF}, graph trend filtering~\eqref{eq:graph_trend_filtering} introduces a graph total variation term to regularize the sparsity of the first-order difference of a graph signal. In this case,  $\x = \h *_v \s = \s$, $u(\cdot) = \alpha \left\| \cdot \right\|_1$, $r(\cdot) = 0$, $\Pj = \Delta$ and $\Q = \Id$.  

Plugging these specifications into~\eqref{eq:GUN_single} leads to a customized graph unrolling network. We consider three modifications. First, we remove the terms related to $\z$ because $r(\z) = 0$. Second, we remove the terms related to $\s$ because $\x = \s$ and there is no need to update both. Third, we replace~\eqref{eq:GUN_y} by a soft-thresholding function which is the analytical solution of~\eqref{eq:general_admm_y}~\cite{BoydPCPE:11}. We finally obtain the $b$th unrolling layer customized for graph trend filtering to be
\begin{subequations}
\label{eq:GUTF}
\begin{eqnarray*}
\label{eq:GUTF_x}
\X^{(b)} & \leftarrow & 
\Bb *_a \T+\Cc *_a \left( \Delta^T \Y^{(b-1)} \right),
\\
\label{eq:GUTF_y}
\Y^{(b)} & \leftarrow & S_{\alpha} \left( \Delta \X^{(b)} \right),
\end{eqnarray*}
\end{subequations}
where $S_{\cdot}(\cdot)$ is a soft-thresholding function. All the training parameters are involved in the edge-weight-sharing  graph convolutions, $\Bb *_a$ and $\Cc *_a$. Since the architecture inherits from the graph unrolling network and graph trend filtering, we call this architecture a~\emph{graph unrolling trend filtering (GUTF)}. The training paradigm follows the general graph unrolling network; see its overall implemetation in Algorithm~\ref{alg:GUTF}. 

Comparing GUTF and GUSC, both follow from the general graph unrolling framework and are based on the proposed edge-weight-sharing graph convolution. The main difference is that GUTF involves a vertex-edge dual representation, where the vertex-based features and edge-based features are converted through the graph incident matrix $\Delta$. This design is potentially better in capturing fast transitions over edges,  leading to improved denoising performance on piecewise-constant graph signals. GUSC heavily relies on the learning ability of the edge-weight-sharing graph convolution, which is potentially more general, but needs more training data than GUTF.

\section{Experimental Results}
\label{sec:experiments}
In this section, we evaluate the proposed graph unrolling networks on denoising both simulated and real-world graph signals. We also test various noise models and various network settings. Our experiments show that the proposed graph unrolling networks consistently achieve better denoising performances than conventional graph signal denoising algorithms and state-of-the-art graph neural networks on both simulated and real-world graph signals under Gaussian noises, mixture noises and Bernoulli noises. We also study the convergence properties of the proposed graph unrolling networks.

\begin{figure*}[thb]
  \begin{center}
    \begin{tabular}{ccc}
   \includegraphics[width=0.47\columnwidth]{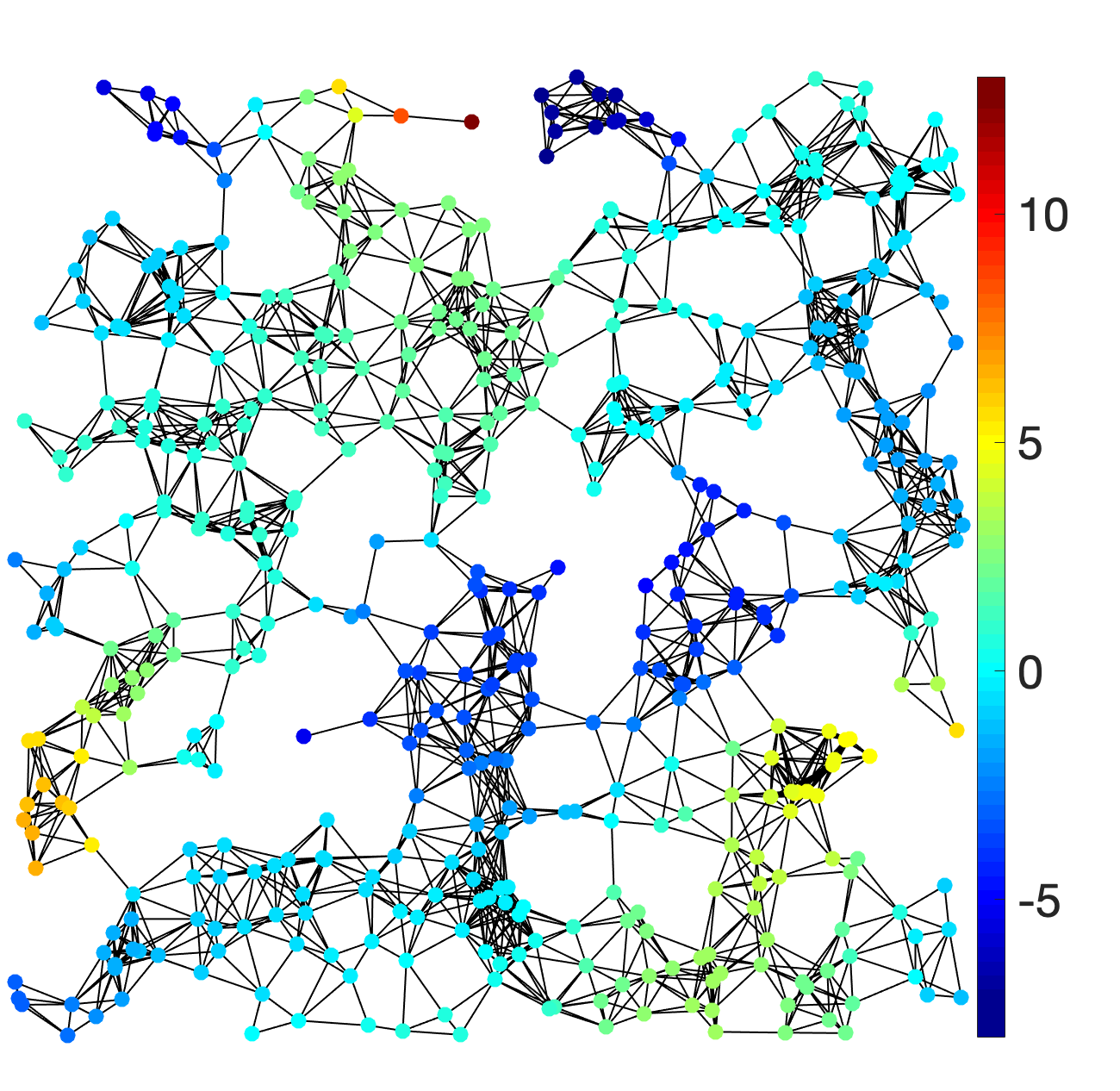}  &   \includegraphics[width=0.47\columnwidth]{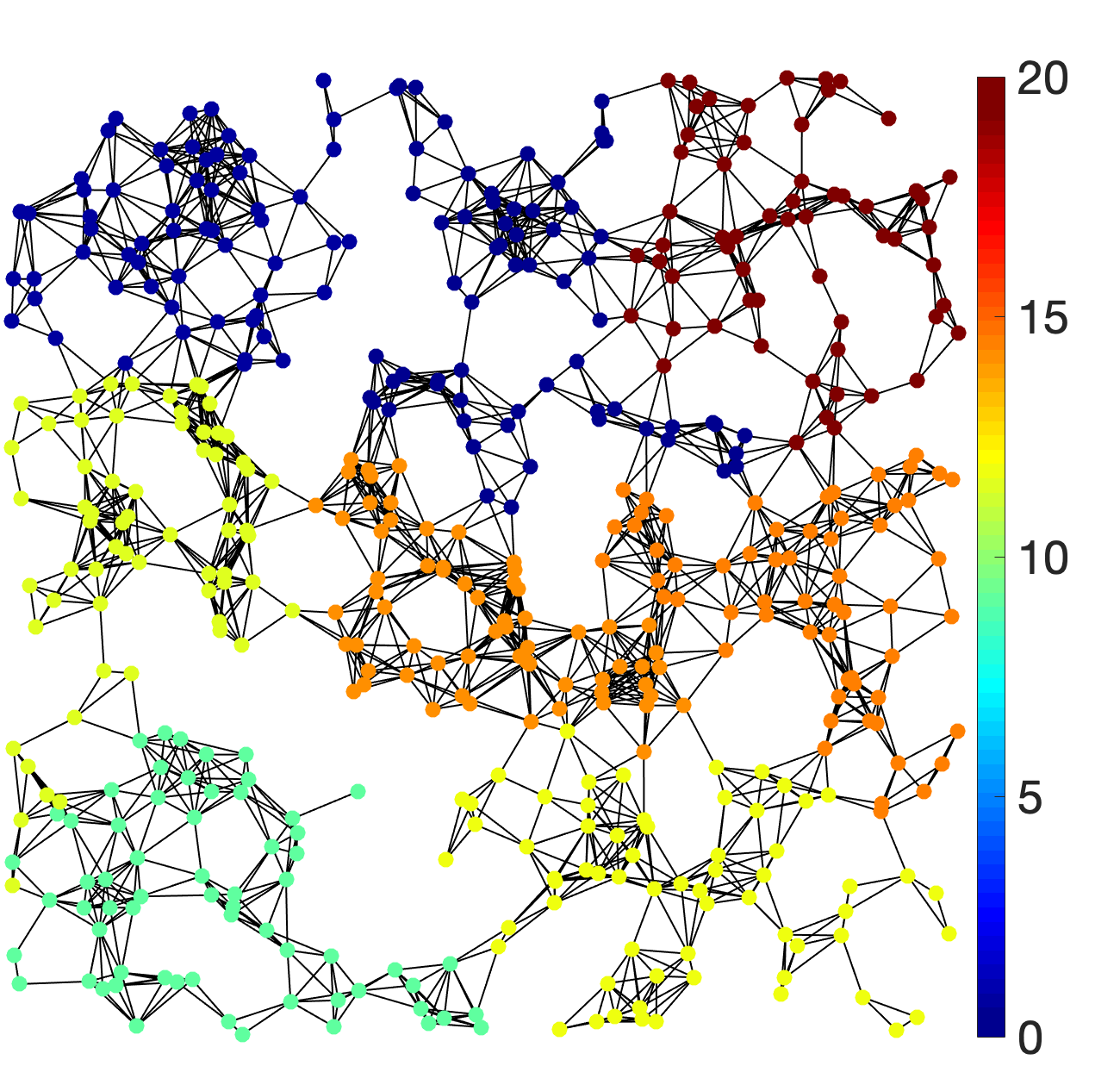}  &  
   \includegraphics[width=0.47\columnwidth]{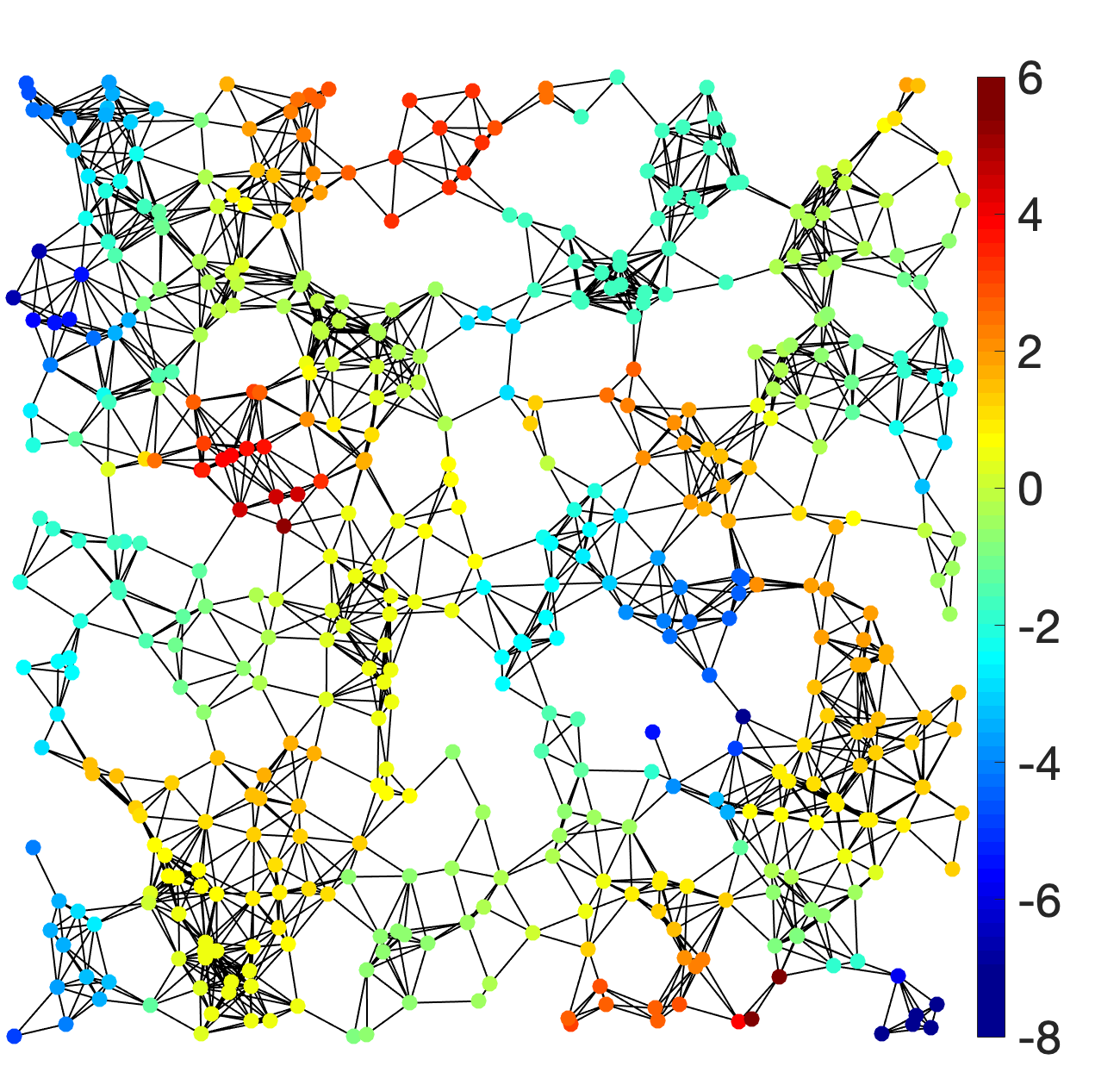}
    \\
    {\small (a) Smooth graph signal.} &  {\small  (b) Piecewise-constant graph signal.}  
    & {\small  (c) Piecewise-smooth graph signal.}

  \end{tabular}
\end{center}
\caption{\label{fig:simulation_example} Visualization of simulation examples.}
\end{figure*}

\begin{table*}[h]
	\vskip 0.1in
	\begin{center}
	   % \scriptsize
		\begin{small}
			\begin{sc}
				\begin{tabular}{c|cccc|cccc|cccc}
					\toprule
					 & \multicolumn{4}{c|}{Smooth} & \multicolumn{4}{c|}{Piecewise-Constant} &
					 \multicolumn{4}{c}{Piecewise-Smooth} \\
					Method & 1 & 10 & 100 & 1000 & 1 & 10 & 100 & 1000 & 1 & 10 & 100 & 1000  \\
					\midrule
					baseline  &  $0.563$ & $0.517$ & $0.498$ & $0.5$ & $0.563$ & $0.516$ & $0.498$ & $0.5$ & $0.562$ & $0.517$ & $0.498$ & $0.5$   \\
					GLD  &  $0.078$ & $0.076$ & $0.071$ & $0.07$ &  $0.045$ & $0.045$ & $0.047$ & $0.042$ &  $0.114$ & $0.111$ & $0.143$	 & $0.13$  \\
					GTF &  $0.098$ &	$0.1$ &	$0.093$ & $0.094$  &  $0.039$ &	$0.043$ &	$0.045$ &	$0.039$ & $0.135$ & $0.115$ & $0.104$ & $0.111$  \\
					GFT &  $0.159$ & $0.125$ &	$0.111$ & $0.112$ &  $0.078$ &	$0.077$ & $0.082$ &	$0.077$ &  $0.162$ & $0.157$ &	$0.185$ & $0.171$ \\
					SGWT & $0.117$ & $0.128$ &	$0.107$ &	$0.110$  & $0.087$ &	$0.086$ & $0.079$	& $0.072$ & $0.172$ &	$0.146$ & $0.153$ &	$0.165$  \\
					QMF & $0.319$ &	$0.322$ & $0.327$ &	$0.334$ & $0.373$ &	$0.369$ &	$0.349$	& $0.346$ & $0.361$ & $0.332$ &	$0.326$ & $0.343$ \\
					CSFB & $0.106$ & $0.101$ &	$0.069$ & $0.075$ & $0.101$ &	$0.094$	 & $0.104$ &	$0.097$ & $0.173$ & $0.163$ & $0.231$ & $0.197$  \\
					\midrule
					MLP &  $0.373$ & $0.189$ & $0.079$ & $0.032$ & $0.182$ & $0.137$ & $0.04$ & $0.014$ & $0.335$ &	$0.209$ & $0.089$ & $0.037$ \\
					GCN & $0.067$ & $0.058$ & $0.039$ & $0.037$ & $0.048$ & $0.039$	& $0.028$ & $0.024$ & $0.102$ & $0.094$ & $0.118$ & $0.074$ \\
					GAT & $0.062$ & $0.057$ & $0.032$ & $0.028$ & ${\bf 0.034}$ & $0.05$ & ${\bf 0.023}$ & $0.018$ & $0.095$ & $0.076$	 & $0.045$	& $0.033$ \\
					\midrule
                    GUSC & $0.049$ & $0.053$ & $0.029$ & ${\bf 0.023}$ & $0.036$ & $0.04$ & $0.024$ & $0.014$ & $0.074$ & $0.069$ & ${\bf 0.031}$ & ${\bf 0.022}$ \\
                    GUTF & ${\bf 0.045}$ & ${\bf 0.046}$ & ${\bf 0.027}$ & ${\bf 0.023}$ & $0.035$ & ${\bf 0.039}$ & ${\bf 0.023}$ & ${\bf 0.011}$ & ${\bf 0.066}$ & ${\bf 0.064}$ & ${\bf 0.031}$ & ${\bf 0.022}$ \\
					\bottomrule
				\end{tabular}
			\end{sc}
		\end{small}
	\end{center}
	\caption{Denoising of three types of simulated graph signals with Gaussian noises. GUTF provides the best denoising performances in most cases.\label{tab:Gaussian_noise}}
\end{table*}

\begin{table*}[h]
	\vskip 0.1in
	\begin{center}
	   % \scriptsize
		\begin{small}
			\begin{sc}
				\begin{tabular}{c|cccc|cccc|cccc}
					\toprule
					 & \multicolumn{4}{c|}{Smooth} & \multicolumn{4}{c|}{Piecewise-Constant} &
					 \multicolumn{4}{c}{Piecewise-Smooth} \\
					Method & 1 & 10 & 100 & 1000 & 1 & 10 & 100 & 1000 & 1 & 10 & 100 & 1000  \\
					\midrule
					baseline  &  $0.393$ & $0.389$ & $0.377$ & $0.374$ & $0.394$ & $0.386$ & $0.377$ & $0.374$ & $0.394$ & $0.388$ & $0.377$	& $0.374$   \\
					GLD  &  $0.05$ & $0.062$ & $0.057$ & $0.057$ & $0.052$ & $0.04$ & $0.041$ & $0.036$ & $0.105$ &  $0.093$ & $0.122$ &	$0.107$  \\
					GTF &  $0.077$ & $0.09$ & $0.08$ &  $0.078$ &		$0.043$ & $0.042$ & $0.039$ & $0.033$ & $0.09$  &	$0.09$ & $0.085$	 &  $0.077$	  \\
					GFT &  $0.096$ & $0.089$ & $0.076$ & $0.077$	&	$0.079$ & $0.067$	& $0.07$ & $0.067$ & $0.137$ &	$0.129$ & $0.159$ & $0.144$  \\
					SGWT & $0.108$ & $0.097$ & $0.092$ & $0.089$ & $0.075$ & $0.074$ & $0.069$ & $0.063$ & $0.133$ & $0.121$ & $0.129$ & $0.141$ \\
					QMF & $0.258$ & $0.277$ & $0.275$ & $0.274$ & $0.29$ & $0.292$ & $0.283$ & $0.277$ & $0.284$ & $0.27$ & $0.266$ & $0.268$	 \\
					CSFB & $0.057$ & $0.079$ & $0.049$ & $0.054$ & $0.115$ & $0.089$ & $0.099$ & $0.095$ & $0.167$ & $0.149$ & $0.216$ & $0.182$	\\
					\midrule
					MLP &  $0.289$ & $0.16$ & $0.057$& $0.026$ & $0.191$ & $0.13$ & $0.028$	& $0.01$ & $0.278$ & $0.167$ & $0.069$	 & $0.03$ \\
					GCN & $0.039$ & $0.048$ & $0.034$ & $0.032$ & $0.051$ & $0.033$ & $0.025$ & $0.021$ & $0.107$ & $0.081$ & $0.1$	 & $0.064$ \\
					GAT &  $0.042$ & $0.043$ & $0.026$ & $0.021$ & $0.045$ & $0.04$ & $0.017$ & $0.007$ & $0.092$ & $0.062$ & $0.036$ & $0.021$\\
					\midrule
                    GUSC &  $0.027$ & $0.044$ & $0.023$ & $0.021$ & $0.043$ & $0.034$ & ${\bf 0.016}$ & ${\bf 0.008}$ & $0.072$ & $0.056$ & $0.03$ & $0.022$ \\
                    GUTF & ${\bf 0.025}$ & ${\bf 0.038}$ & ${\bf 0.022}$ &	${\bf 0.018}$	& ${\bf 0.041}$ & ${\bf 0.031}$	& ${\bf 0.016}$ &	${\bf 0.008}$ & ${\bf 0.06}$ & ${\bf 0.05}$ & ${\bf 0.026}$ & ${\bf 0.018}$\\
					\bottomrule
				\end{tabular}
			\end{sc}
		\end{small}
	\end{center}
	\caption{Denoising of three types of simulated graph signals with mixture noises (Gaussian and Laplace). GUTF provides the best denoising performances in most cases.	\label{tab:mixture_noise}}
\end{table*}

\subsection{Experimental setup}
\mypar{Configurations} For GUSC and GUTF, we set the number of epochs for stochastic gradient descent to be $5000$, the number of network layers $B = 1$, feature dimension $d^{(b)} = D^{(b)} = 64$, the threshold in the soft-thresholding function $\alpha = 0.05$ in all the cases. To make a fair comparison, we use the same network setting and training paradigm for graph unrolling networks to train other networks.

\mypar{Baselines}
We consider three classes of competitive denoising algorithms: graph-regularized optimizations, graph filter banks and neural networks. For graph-regularized optimizations, we select graph Laplacian denoising (GLD)~\cite{ShumanNFOV:13} and graph trend filtering (GTF)~\cite{WangSST:16}. Both algorithms introduce graph-regularization  terms to the optimization problem. For graph filter banks, we consider graph Fourier transform (GFT) ~\cite{ShumanNFOV:13}, spectral graph wavelet transform (SGWT)~\cite{HammondVG:11}, graph quadrature-mirror-filters (QMF)~\cite{NarangO:12} and critically sampled filter banks (CSFB)~\cite{TremblayB:16}.
In each case, we obtain the corresponding graph dictionary and use basis pursuit denoising~\cite{BoydPCPE:11} to reconstruct graph signals from noisy inputs. As competative neural networks, we consider multilayer perception with three fully-connected layers~\cite{Goodfellow:2016}, graph convolution networks (GCN) with three graph convolution layers~\cite{KipfW:17}, graph attention networks (GAT) with one graph attention layer~\cite{VelivckovicCCRLB:18} and graph autoencoder (GAE) with three graph convolution layers and one kron-reduction pooling layer~\cite{DoND:10}. In many supervised-learning tasks, previous works realized that  deep graph neural networks with too many layers can suffer from oversmoothing and hurt the overall performance~\cite{WuSZFYW:19,ZhaoA:20}. In our experiments, we also find that more layers do not lead to better denoising performance even with residual connections. We tune hyperparameters for each denoising algorithm and report the best performances.

\mypar{Graph signals} We consider three types of simulated graph signals as well as three types of real-world graph signals. For simulations, we consider smooth, piecewise-constant and piecewise-smooth graph signals; for real-world scenarios, we consider temperature data supported on the U.S weather stations,
traffic data based on the NYC street networks and community memberships based on citation networks. The details will be elaborated in each case.

\mypar{Noise models}
We consider three types of noises to validate the denoising algorithms: Gaussian noise, the mixture noise and Bernoulli noise. In the measurement model~\eqref{eq:noise_model}, we use a length-$N$ vector $\ee$ to denote noise. For  Gaussian noise, each element of $\ee$ follows a Gaussian distribution with zero mean; that is, $\ee_i \sim  \mathcal{N}(0, \sigma^2)$. The default standard deviation is $\sigma = 0.5$. For the mixure noise, each element of $\ee$ follows a mixture of Gaussian distribution and Laplace distribution; that is, $\ee_i \sim  \mathcal{N}(0, \sigma^2) + Laplace(0, b)$. By default, we set $\sigma = 0.2, b = 0.2$. For binary graph signals, we consider adding Bernoulli noise~\cite{ChenYZSK:17}; that is, we randomly select a subset of vertices and flip the associated binary values. Note that (i) the proposed graph unrolling network is not designed for this noise model, but surprisingly, it still performs well; (ii) we change the loss function~\eqref{eq:loss} to the cross-entropy loss during training; and (iii) this denoising task is essentially a classification task, identifying whether the binary value at each vertex is flipped.

\mypar{Evaluation metrics}
To evaluate the denoising performance, the default metric is the normalized mean square error (NMSE); that is,
\begin{equation*}
 {\rm NMSE} \ = \ \frac{\left\| \widehat{\x} - \x  \right\|_2^2}{\left\| \x  \right\|_2^2},
\end{equation*}
where $\x \in \R^{N}$ is a noiseless graph signal, and $\widehat{\x}$ is a denoised graph signal. A smaller value of NMSE indicates a better denoising performance. We also consider the normalized mean absolute error (NMAE); that is,
\begin{equation*}
 {\rm NMAE} \ = \ \frac{\left\| \widehat{\x} - \x  \right\|_1}{\left\| \x  \right\|_1}.
\end{equation*}

For binary graph signals, we evaluate by the error rate (ER),
\begin{equation*}
 {\rm ER} \ = \  \frac{1}{N} \sum_{i=1}^N {\bf 1}{\left( x_i \neq \widehat{x}_i \right)},
\end{equation*}
where $x_i$ and $\widehat{x}_i$ are the $i$th element in $\x, \widehat{\x}$, respectively. A smaller value of ER indicates a better denoising performance. We also consider the F1 score, which is the harmonic mean of the precision and recall. A higher value of F1 indicates a better denoising performance.

\begin{table*}[h]
	\vskip 0.1in
	\begin{center}
	   % \scriptsize
		\begin{small}
			\begin{sc}	\begin{tabular}{c||cc||cc||cc|cc}
					\toprule
					& \multicolumn{2}{c||}{Temperature} &  \multicolumn{2}{c||}{Traffic}
					&  \multicolumn{4}{c}{cora}\\
					Metric & \multicolumn{2}{c||}{NMSE}  & \multicolumn{2}{c||}{NMSE} 
					 & \multicolumn{2}{c|}{Error Rate} 
					 & \multicolumn{2}{c}{F1 score} \\
					Method & 1 & 365  & 1 & 24  & 1 & 7 & 1 & 7 \\
					\midrule
					Baseline & $0.34$  & $0.377$ & $0.392$ & $0.377$ & $0.095$ & $0.099$ & $0.829$ & $0.695$ \\
					GLD  &  $0.045$ & $0.024$ & $0.248$ & $0.255$ & $0.055$ & $0.032$ & $0.609$ & $0.793$ \\
					GTF & $0.079$  & $0.036$ & $0.202$ & $0.176$ &   $0.06$ & $0.039$ &$0.484$  & $0.532$\\
					GFT & $0.065$   & $0.053$ & $0.257$ & $0.231$ &  $0.079$ & $0.053$ & $0.459$ & $0.489$ \\
					SGWT & $0.069$ & $0.11$ & $0.184$ & $0.162$ &  $0.074$ & $0.073$ & $0.551$ & $0.569$ \\
					QMF &  $0.26$ & $0.31$ & ${\bf 0.18}$ & $0.185$ & 	$0.087$ & $0.087$ & $0.51$ & $0.512$ \\
					CSFB & $0.07$  & $0.061$ & $0.344$ & $0.36$ & $0.129$ & $0.143$ & $0.43$ & $0.437$	 \\
					\midrule
					MLP & $0.142$ & $0.02$7 & $0.31$ & $0.169$  & $0.095$ & $0.072$ & $0.829$ & $0.695$ \\
					GCN & $0.041$ & $0.033$ & $0.293$ & $0.279$  & $0.042$ & $0.025$ & $ 0.903$ & $ 0.901$ \\
					GAT & $0.044$   & $0.031$ & $0.267$ & $0.264$ & $0.041$ & $0.032$ & $0.909$ & $0.873$ \\	
					\midrule
					GUSC & ${\bf 0.037}$  & ${\bf 0.016}$ & $0.324$  & $0.178$ & ${\bf 0.04}$ & ${\bf 0.024}$ & ${\bf 0.91}$ & ${\bf 0.906}$ \\
					GUTF & $0.053$  &  $0.019$ & $0.266$ & ${\bf 0.158}$ & $0.041$ & $0.03$ & $0.906$ & $0.885$ \\
					\bottomrule
				\end{tabular}
			\end{sc}
		\end{small}
	\end{center}
	\caption{Denoising of real-world data with mixture noises (Gaussian and Laplace). GUSC produces the best denoising performances in most cases.	\label{tab:real_world_denoise}}
	\vspace{-5mm}
\end{table*}

\subsection{Simulation validation}
\label{sec:simulation}
\mypar{Smooth graph signals} 
We simulate a random geometric graph, by generating an undirected graph with $500$ vertices randomly sampled from the unit square. Two vertices are connected when their Euclidean distance is less than a threshold. To generate a smooth graph signal, we consider the  bandlimited graph signals~\cite{ChenVSK:15}. We conduct the eigendecomposition of the graph Laplacian matrix and ascendingly order the eigenvalues. The first few eigenvectors span a subspace of smooth graph signals, called a bandlimited space~\cite{ChenVSK:15}. We use the a linear combination of the first few eigenvectors to obtain a smooth graph signal; see  Fig.~\ref{fig:simulation_example} (a).

We denoise four different numbers of graph signals: $1, 10, 100$ and $1000$. We expect that graph unrolling networks will provide better performances with more samples. Columns $2 - 5$ in Tables~\ref{tab:Gaussian_noise} and~\ref{tab:mixture_noise} compare the denoising performances of smooth graph signals under Gaussian noise and the mixure noise, respectively.  We see that (i) the proposed two graph unrolling networks significantly outperforms all the other competitive methods. For denoising a single graph signal, GUTF is around $40\%$ better than the standard graph Laplacian denoising; for denoising $1000$ graph signal, GUTF is around $70\%$ better than the standard graph Laplacian denoising! (ii) among the conventional graph signal denoising algorithms, graph Laplacian denoising achieves the best performances;  (iii) neural-network-based methods overall outperform conventional graph signal denoising algorithms. Surprisingly, even training with a single graph signal, most neural networks still provide excellent denoising performance; and (iv) as standard neural networks, MLP performs poorly when training samples are few and gets better when the number of training samples is increased, which makes sense because MLP does not leverage any graph structure. This shows that we cannot expect an arbitrary neural network without dedicated design to work well for graph signal denoising. % Qualitatively, Figure~\ref{fig:simulation_example} (d) shows a noisy smooth graph signal and (g) shows the denoised output by GUTF.

\mypar{Piecewise-constant graph signals} 
We next simulate piecewise-constant graph signals on a random geometric graph. We first randomly partition the graph into a fixed number of connected and mutually exclusive subgraphs with roughly the same size. Within each subgraph, for each graph signal, we randomly generate a constant value over all vertices in the subgraph. The generated graph signal is piecewise-constant and only changing at the boundary between graph partitions; see an example in Fig.~\ref{fig:simulation_example} (b).

Again, we denoise four different numbers of graph signals: $1, 10, 100$ and $1000$.  Columns $6 - 9$ Tables~\ref{tab:Gaussian_noise} and~\ref{tab:mixture_noise} compare the denoising performances of piecewise-constant graph signals under Gaussian noise and the mixture noise, respectively. Similar to smooth graph signals, we see that (i) the proposed two graph unrolling networks still significantly outperform all the other competitive methods; ii) among the conventional graph signal denoising algorithms, graph trend filtering achieves the best performances as its graph regularization promotes piecewise-constant graph signals; and (iii) MLP fails with few training samples. % Qualitatively, Figure~\ref{fig:simulation_example} (e) shows a noisy piecewise-constant graph signal and (h) shows the denoised output by GUTF.

\mypar{Piecewise-smooth graph signals} 
We simulate piecewise-smooth graph signals on a random geometric graph. Similar to piecewise-constant signals, we first partition the graph into mutually exclusive subgraphs. Within each subgraph we generate smoothing signals based on the first-$K$ eigenvectors of the subgraph's Laplacian matrix, using the same approach as generating smooth graph signals. The combined signal over the whole graph is piecewise-smooth; see an example in Fig.~\ref{fig:simulation_example} (c).

Columns $10 - 12$ Tables~\ref{tab:Gaussian_noise} and~\ref{tab:mixture_noise} compare the denoising performances of piecewise-smooth graph signals under Gaussian noise and the mixture noise, respectively. Similar to smooth graph signals, we see that (i) the proposed two graph unrolling networks still significantly outperform all the other competitive methods; ii) among the conventional graph signal denoising algorithms, graph trend filtering achieves the best performances as its graph regularization promotes piecewise-constant graph signals; and (iii) MLP fails with few training samples. % Qualitatively, Figure~\ref{fig:simulation_example} (f) shows a noisy piecewise-smooth graph signal and (i) shows the denoised output by GTUF.

\mypar{Influence of noise level}
To validate the effect of noises, we vary the noise level and compare the denoising performances of graph neural networks, including GCN, GAT, GAT, GUSC and GUTF. Here we consider Gaussian noises and the noise level is the standard deviation of the noise. Fig.~\ref{fig:GS_noise} and~\ref{fig:PS_noise} compare the denoising performances of smooth graph signals and piecewise-smooth graph signals as a function of noise level, respectively, where the $x$-axis is the noise level and $y$-axis is the logarithm of NMSE and NMAE. We see that the proposed GUSC and GUTF consistently outperform the other methods across all noise levels.

\begin{figure}[thb]
  \begin{center}
    \begin{tabular}{cc}
   \includegraphics[width=0.47\columnwidth]{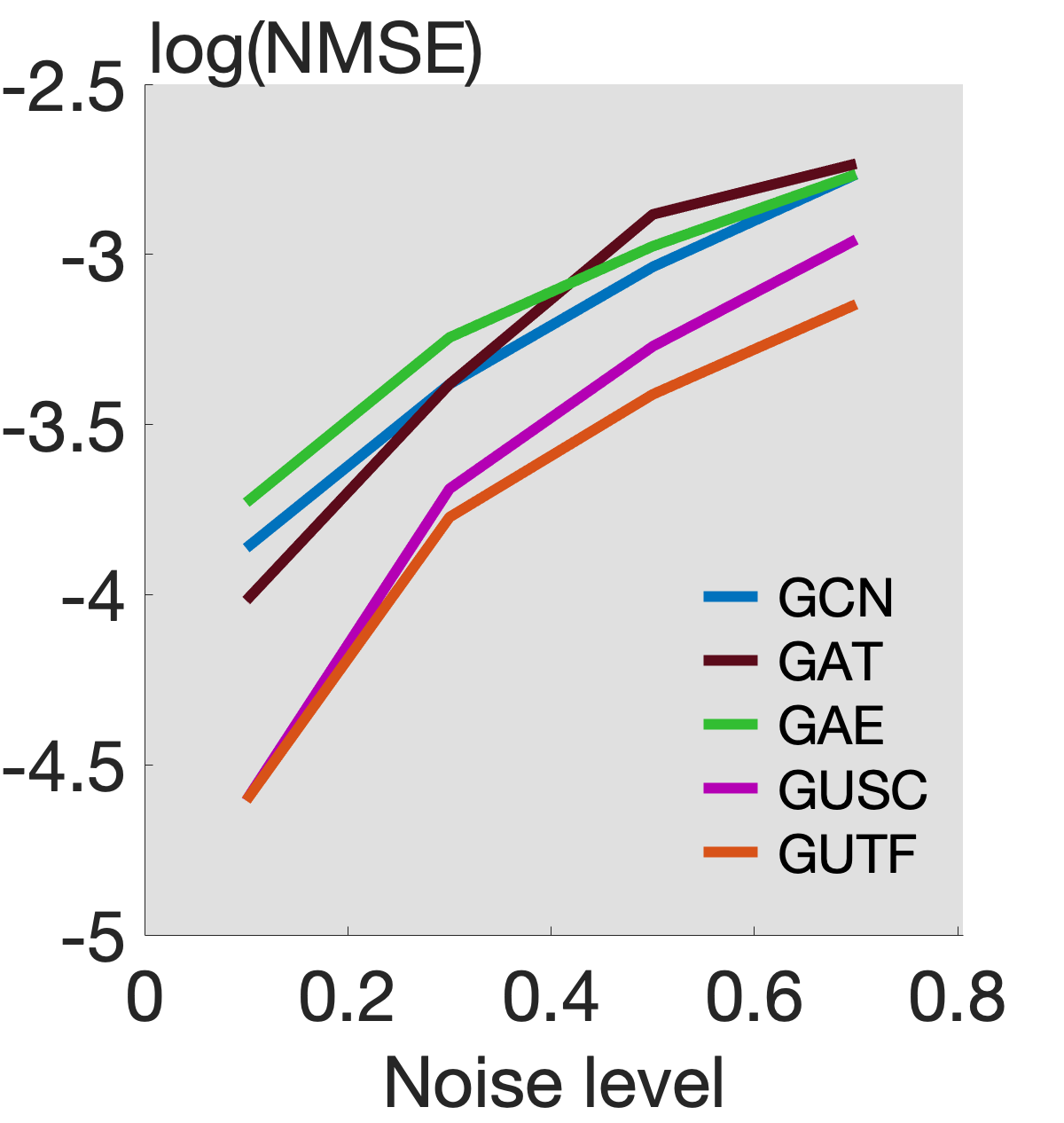} &      \includegraphics[width=0.47\columnwidth]{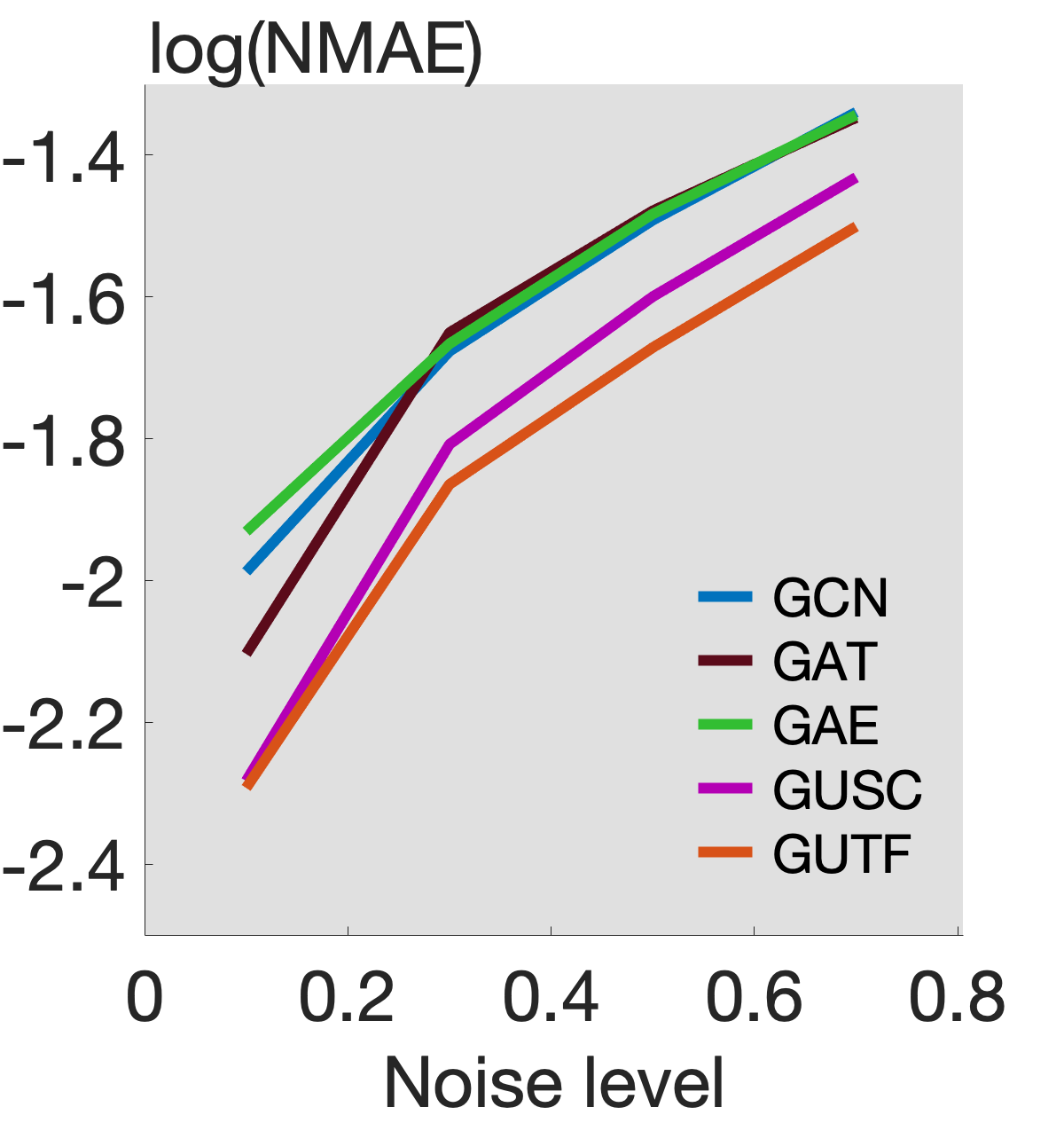}
    \\
    {\small (a) NMSE.} &  {\small  (b) NMAE.}
  \end{tabular}
\end{center}
\caption{\label{fig:GS_noise} Denoising performance of smooth graph signals as a function of noise level. GTUF provides the best denoising performance under both metrics.}
\end{figure}

\begin{figure}[thb]
  \begin{center}
    \begin{tabular}{cc}
   \includegraphics[width=0.47\columnwidth]{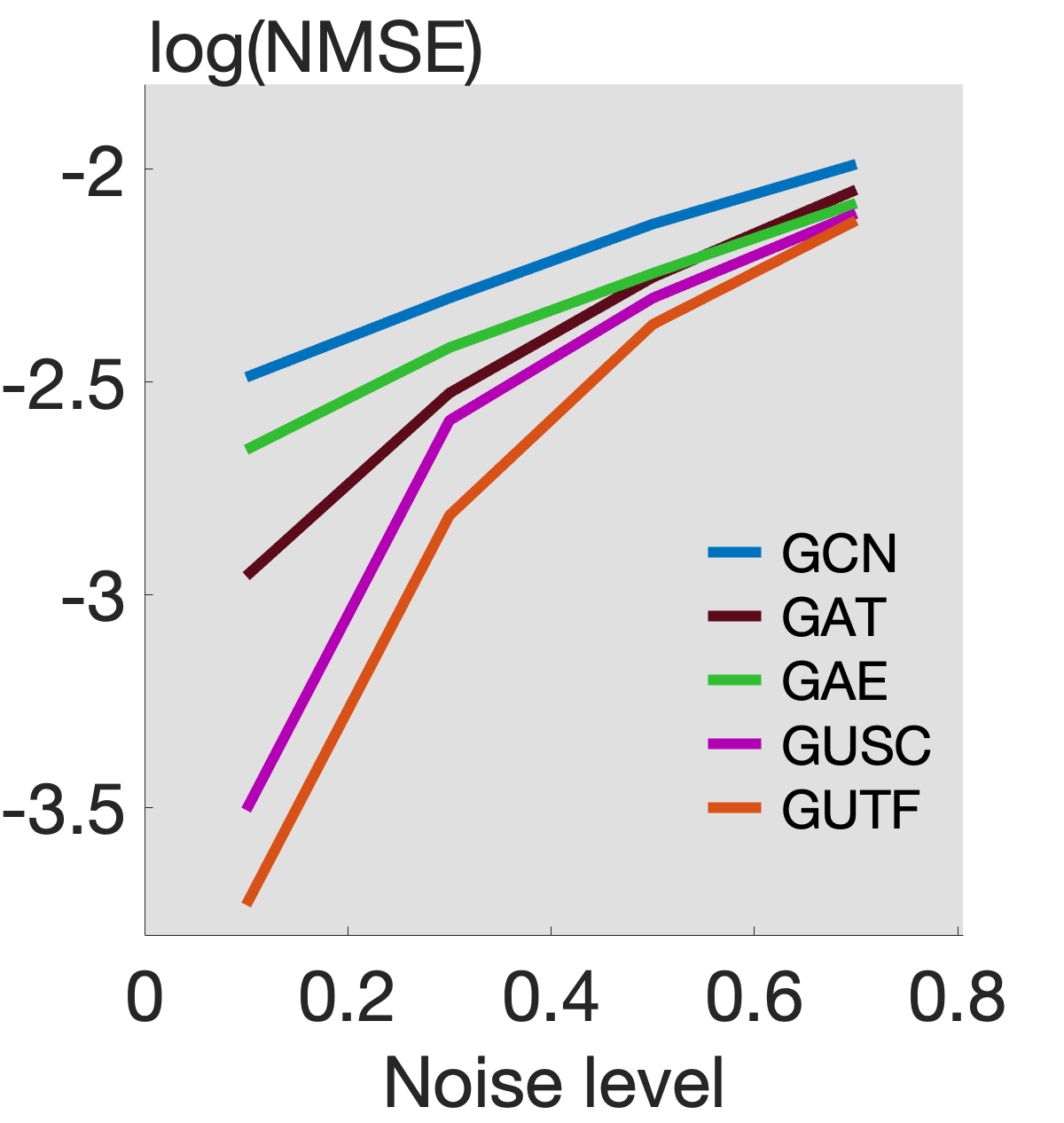} &      \includegraphics[width=0.47\columnwidth]{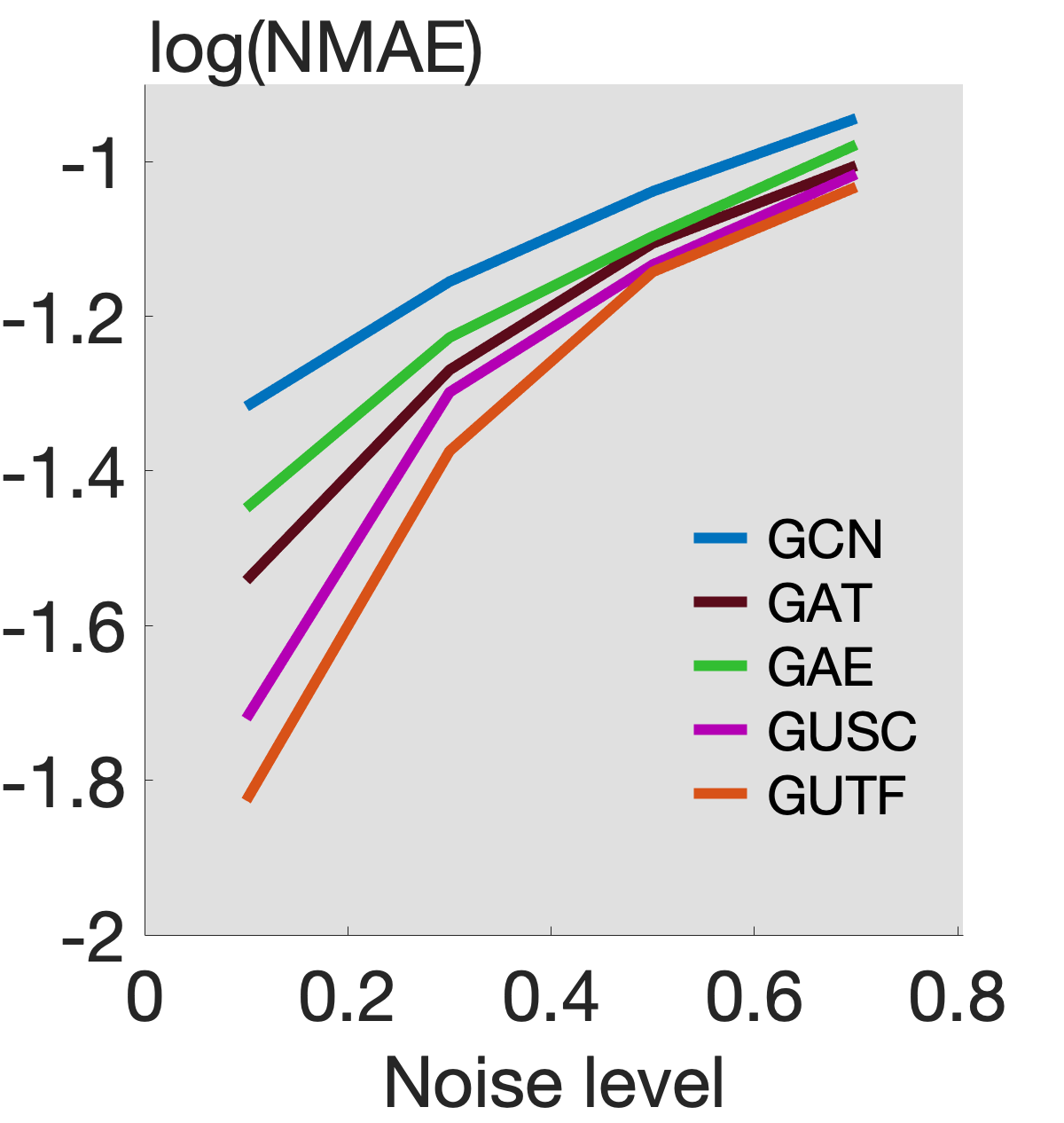}
    \\
    {\small (a) NMSE.} &  {\small  (b) NMAE.}
  \end{tabular}
\end{center}
\caption{\label{fig:PS_noise} Denoising performance of piecewise-smooth graph signals as a function of noise level. GTUF provides the best denoising performance under both metrics.}
\end{figure}

\begin{figure}[thb]
  \begin{center}
    \begin{tabular}{cc}
   \includegraphics[width=0.47\columnwidth]{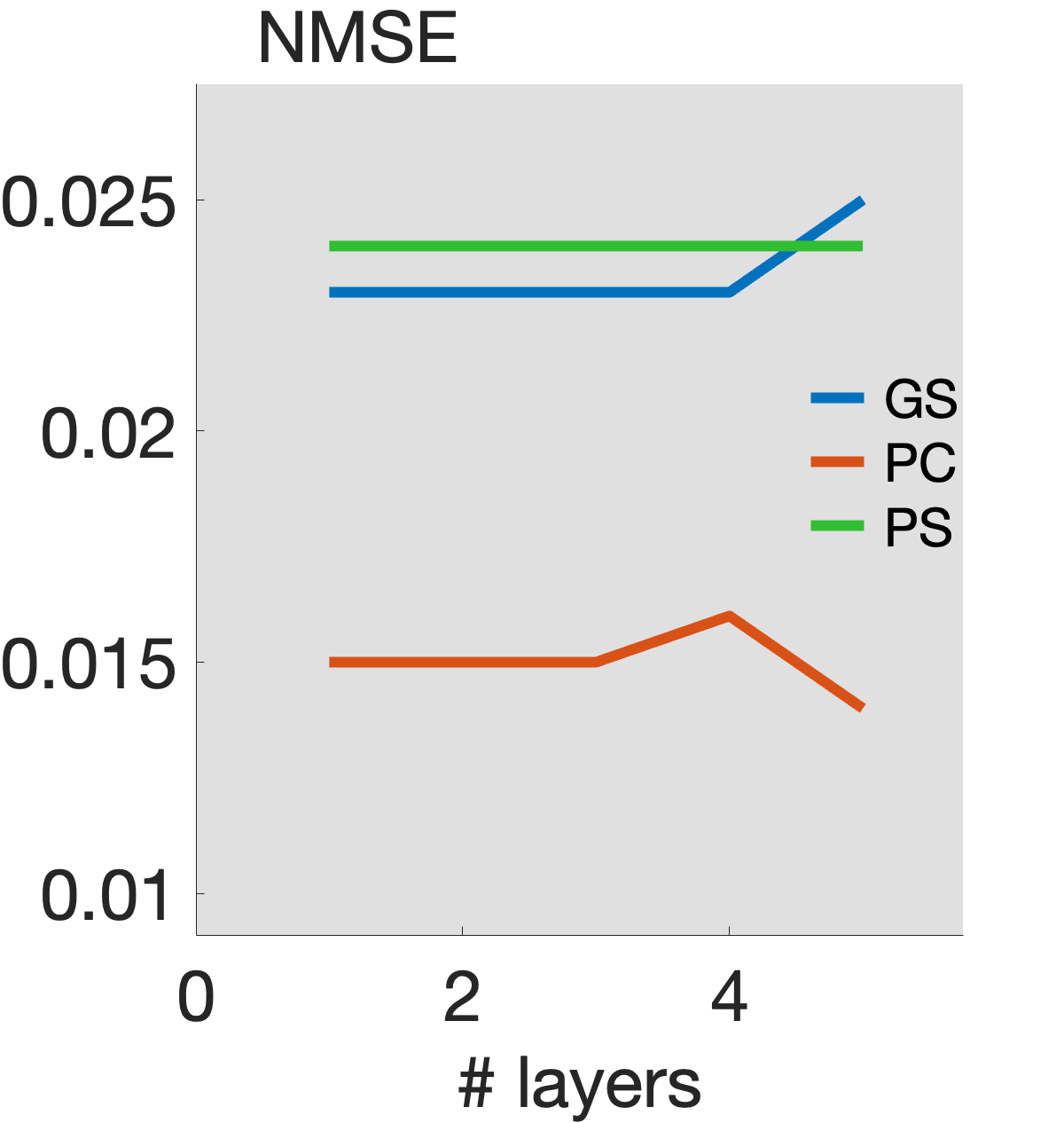} &      \includegraphics[width=0.47\columnwidth]{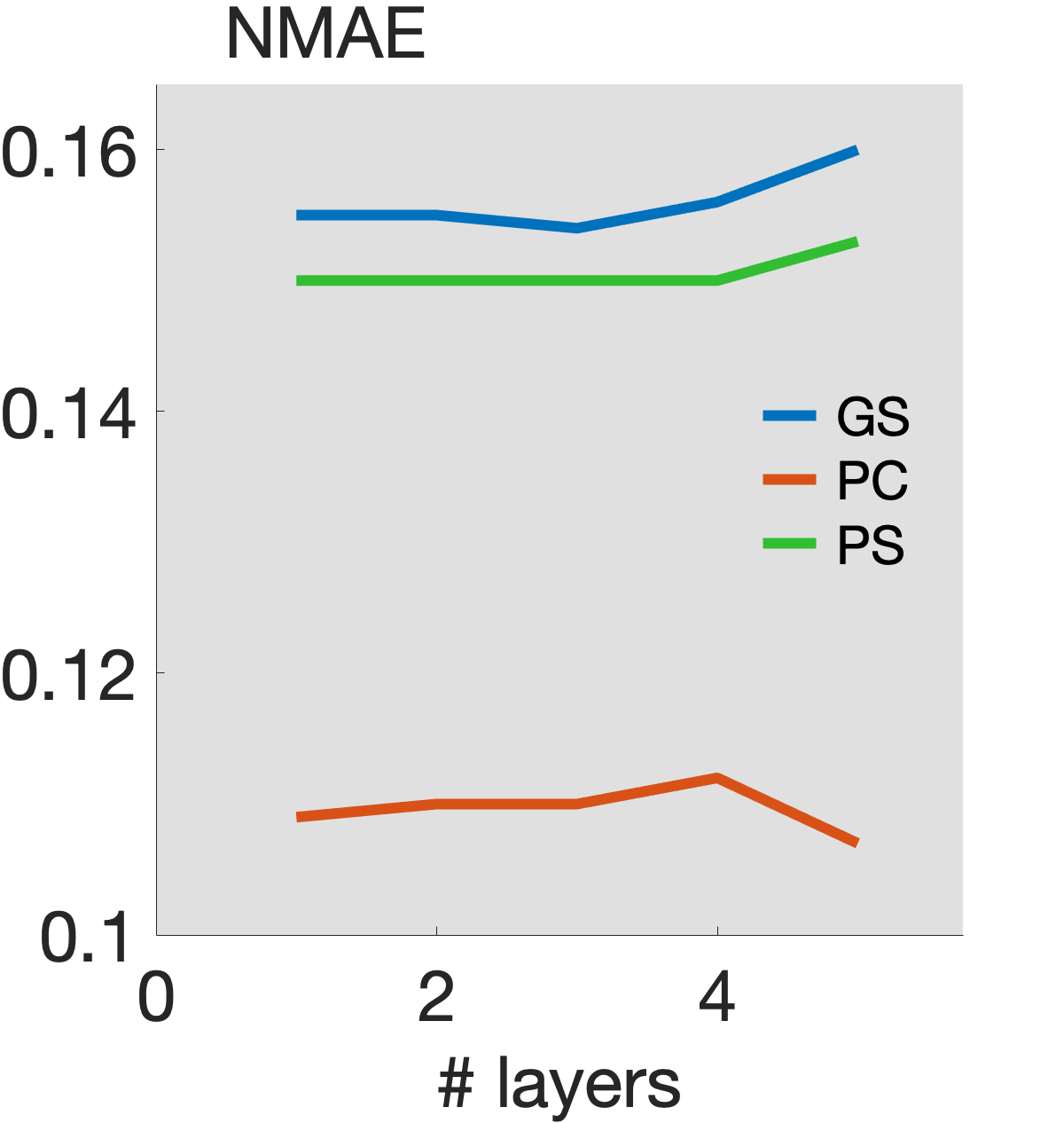}
    \\
    {\small (a) NMSE.} &  {\small  (b) NMAE.}
  \end{tabular}
\end{center}
\caption{\label{fig:num_layers} Denoising performance as a function of number of layers. Varying the number of layers has litter effect on the denoising performance.}
\end{figure}
 
\mypar{Influence of number of layers}
To validate the effect of the number of layers, we vary the number of layers from $1$ to $5$ for GUSC and show the denoising performances on smooth graph signals, piecewise-constant graph signals and piecewise-smooth graph signals. Fig.~\ref{fig:num_layers} shows the denoising performance of three types of graph signals as a function of the number of layers, where the $x$-axis is the number of layers and the $y$-axis is NMSE for Plot (a) and NMAE for Plot (a). We see that varying the number of layers has little effect on the denoising performance.

\begin{figure}[thb]
  \begin{center}
    \begin{tabular}{cc}
   \includegraphics[width=0.47\columnwidth]{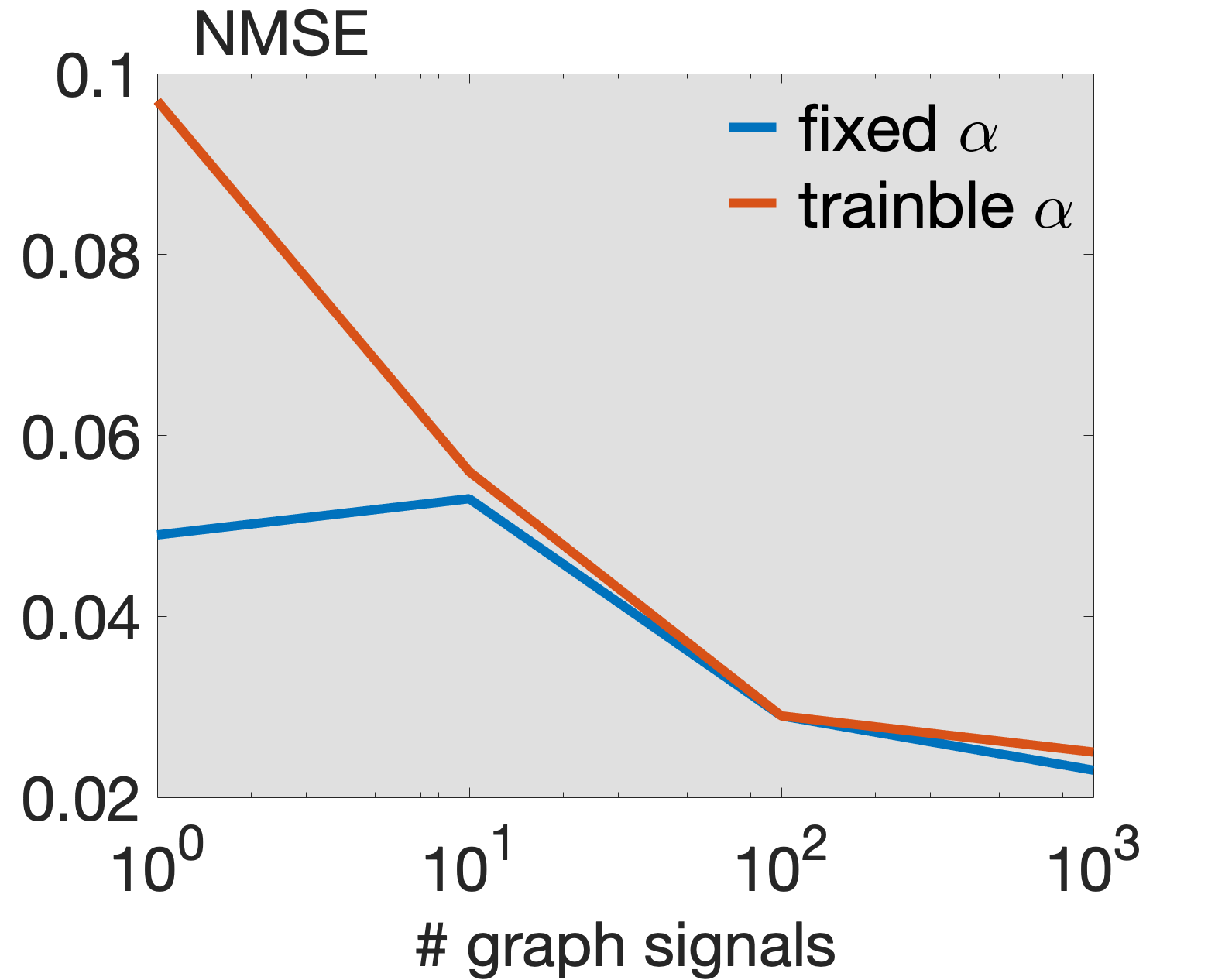} &         \includegraphics[width=0.47\columnwidth]{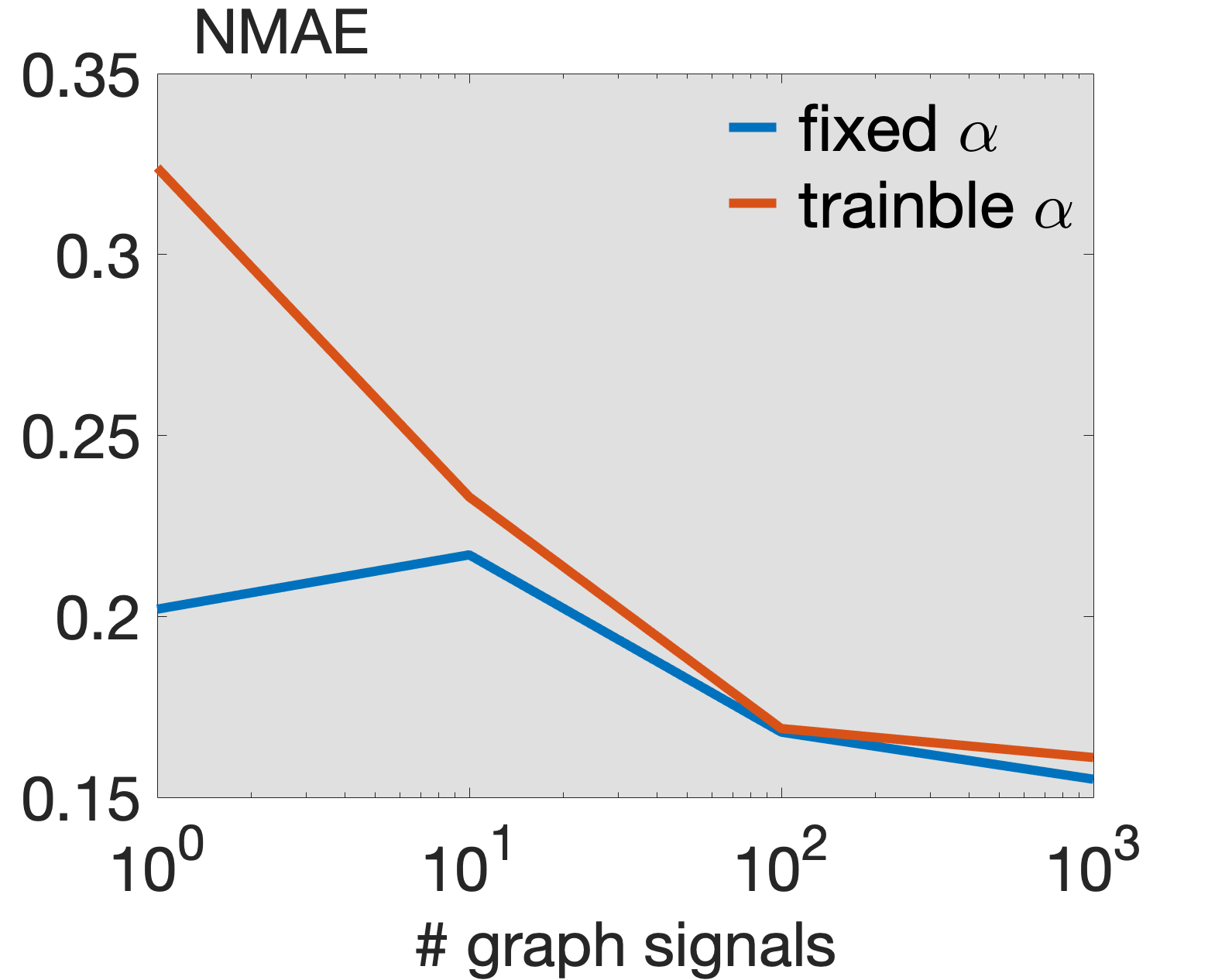}
    \\
    {\small (a) NMSE.} &  {\small  (b) NMAE.}
  \end{tabular}
\end{center}
\caption{\label{fig:threshold} Comparison of denoising performance with a fixed threshold $\alpha$ and a trainable threshold $\alpha$.}
\end{figure}

\mypar{Influence of threshold $\alpha$}
We next compare the denoising performances with a fixed threshold $\alpha = 0.05$ and a trainable threshold $\alpha$ in each individual layer. Fig.~\ref{fig:threshold} compares the denoising performances of smooth graph signals as a function of the number of graph signals, where the $x$-axis is the number of graph signals and the $y$-axis is NMSE for Plot (a) and NMAE for Plot (b). 
For denoising a few graph signals, a fixed threshold leads to a better performance. When we increase the number of graph signals, the gap between a fixed threshold and a trainable threshold becomes negligable.

\mypar{Convergence analysis}
To validate the convergence of the proposed networks, we record the network output at each epoch and compare it with noisy graph signals and clean graph signals in terms of NMSE. Fig.~\ref{fig:convergence} shows the logarithm-scale NMSE as a function of the number of epochs for a smooth graph signal and a piecewise-smooth graph signal, respectively. In each plot, the blue curve shows the difference between the denoised graph signal and the noisy graph signal; and the red curve shows the difference between the denoised graph signal and the clean graph signal. We see that, even when we use the noisy graph signal as the supervision to train the network, the denoised output is much closer to the clean graph signal, which is always unknown to the network. Specifically for the blue curve, the loss goes down quickly in the beginning because the network easily captures the graph-structure-related component from  the noisy graph signal. Later on, the loss does not continue to drop as the network does not fit the noisy component, leading to a clean graph signal. This reflects the implicit graph regularization of the proposed network.

\begin{figure}[thb]
  \begin{center}
    \begin{tabular}{cc}
   \includegraphics[width=0.47\columnwidth]{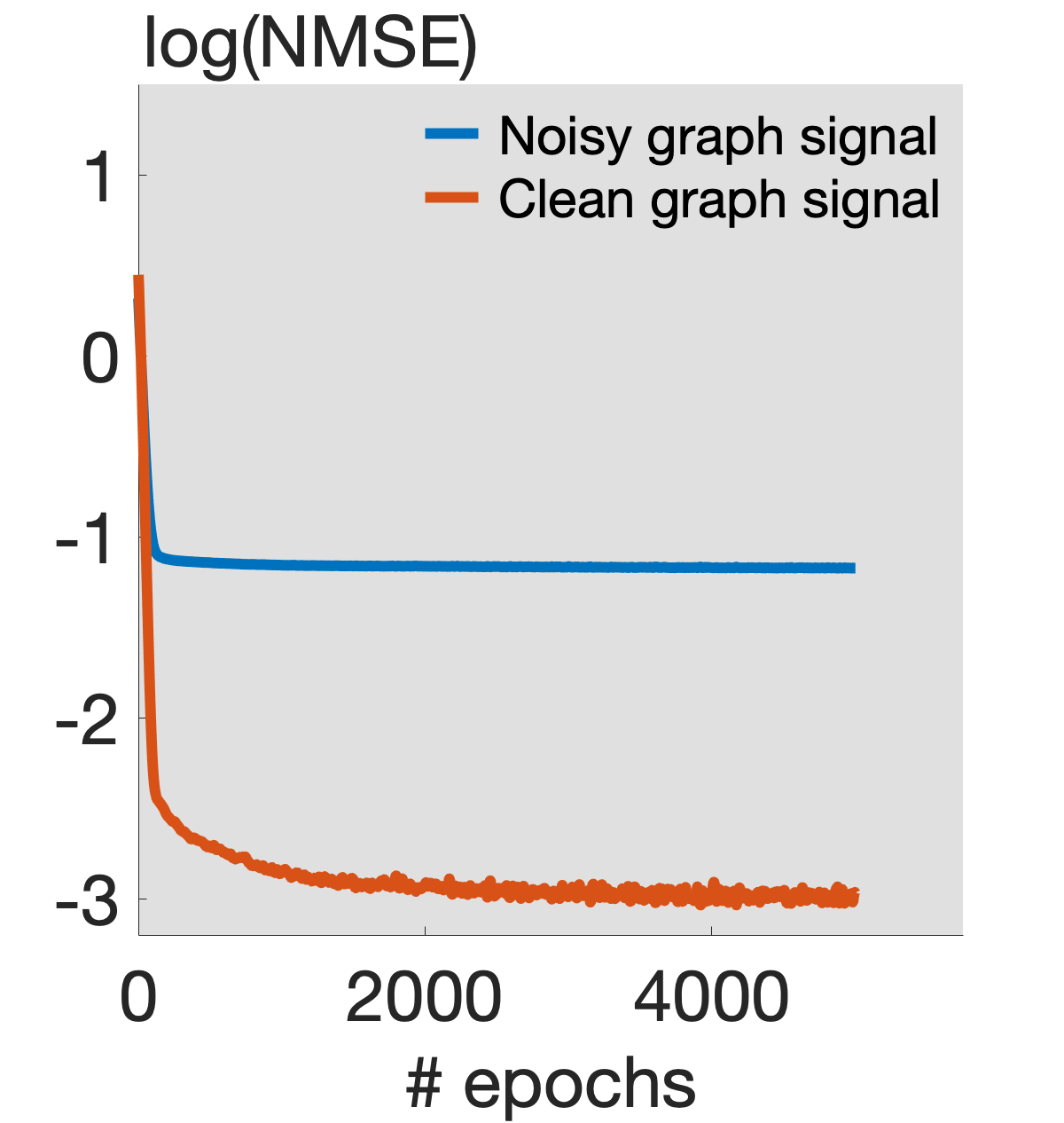} &      \includegraphics[width=0.47\columnwidth]{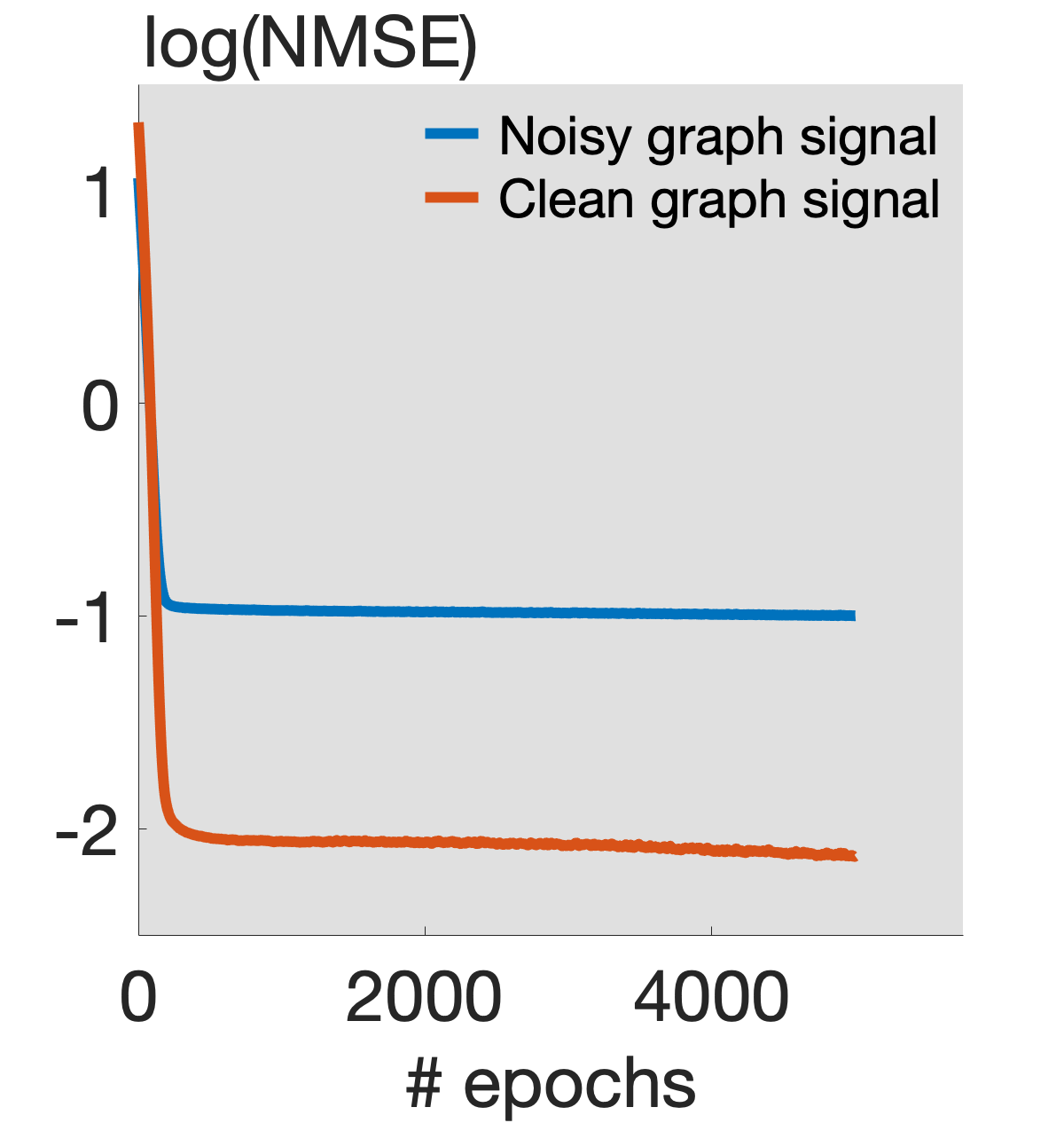}
    \\
    {\small (a) Smooth.} &  {\small  (b) Piecewise-smooth.}
  \end{tabular}
\end{center}
\caption{\label{fig:convergence} NMSE between the denoised graph signal and the target as a function of the number of epochs. The target is either a noisy graph signal (marked in blue) or a clean graph signal (marked in red). For the blue curve, the loss goes down quickly in the beginning because the network easily captures the graph-structure-related component from  the noisy graph signal. Later on, the loss does not continue to drop as the network does not fit the noisy component, leading to a clean graph signal.
}
\end{figure}

\begin{figure}[thb]
  \begin{center}
    \begin{tabular}{cc}
   \includegraphics[width=0.45\columnwidth]{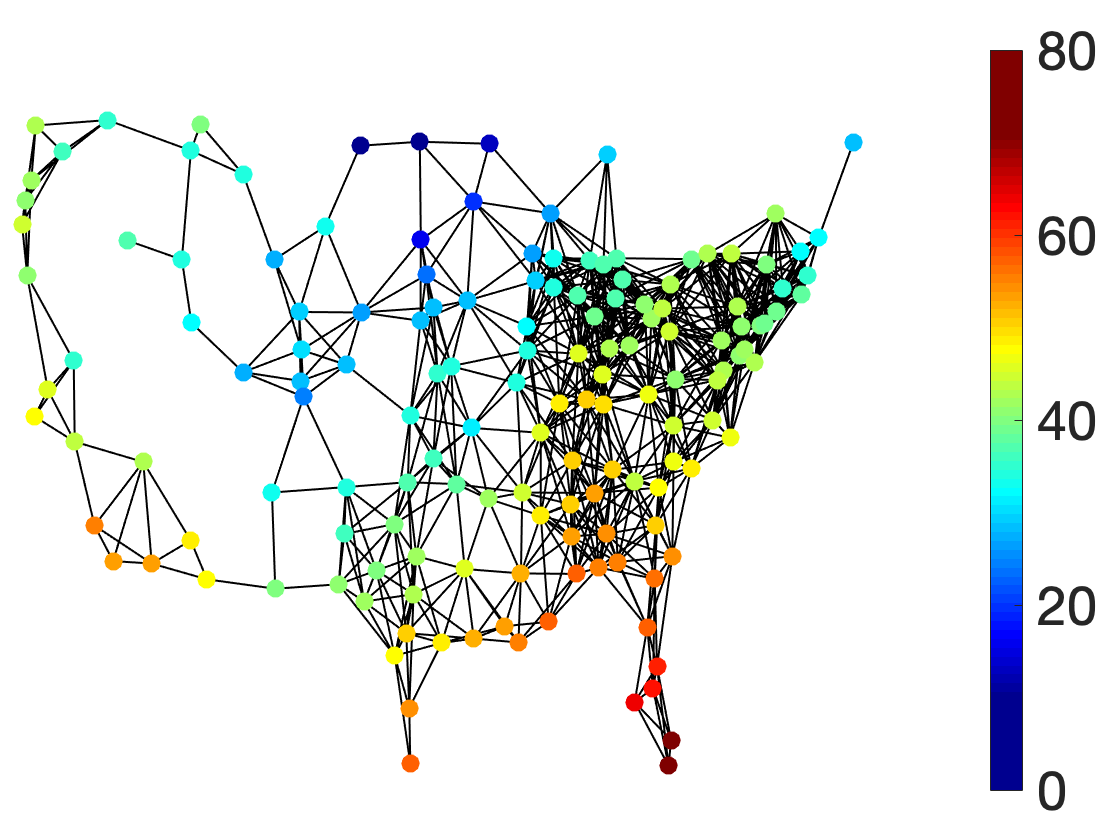} &      \includegraphics[width=0.5\columnwidth]{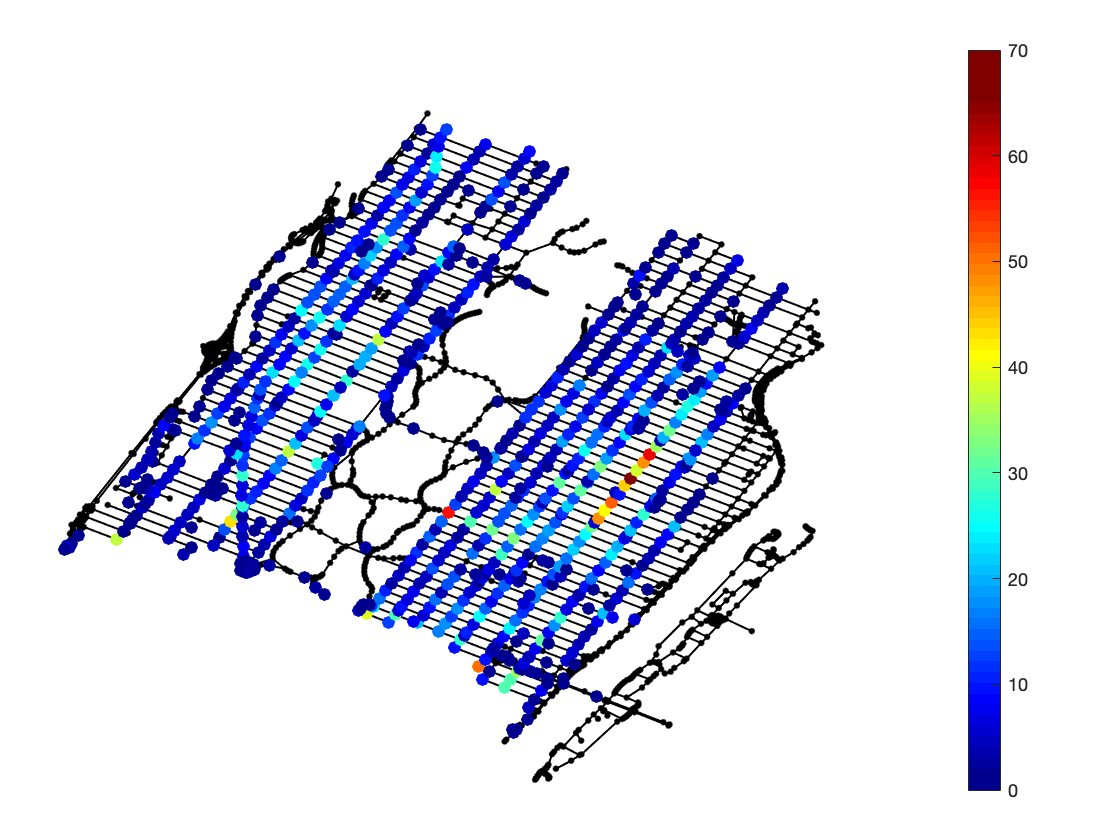}
    \\
    {\small (a) Temperature in the U.S.} &  {\small  (b) Traffic data in Manhattan.}
  \end{tabular}
\end{center}
\caption{\label{fig:real_world_example} Visualization of of real-world examples.}
\end{figure}

\subsection{Real-world examples}
\mypar{U.S. temperature data} We consider 150 weather stations in the United States that record
their local temperatures~\cite{SandryhailaM:13}. Each weather station has 365 days of
recordings (one recording per day), for a total of 54,750 measurements.  The graph
representing these weather stations is obtained by measuring the
geodesic distance between each pair of weather stations. The vertices are
represented by an $8$-nearest neighbor graph, in which vertices represent
weather stations, and each vertex is connected to eight other vertices that
represent the eight closest weather stations. Each graph signal is the daily temperature values recorded in each weather station; see one example in Fig.~\ref{fig:real_world_example} (a). We have $365$ graph signals in total. Intuitively, those graph signals are smooth over the underlying graph because neighboring weather stations record similar temperatures.

We denoise two different numbers of graph signals: $1$ and $365$. Columns $2 - 3$ in Tables~\ref{tab:real_world_denoise} compare the denoising performances under the mixture noise. We see that i) the proposed GUSC significantly outperforms all the other competitive methods in terms of NMSE; ii) the proposed GUTF does not work well for a single graph signal, but performs well when more training data is given; and iii) The graph Laplacian denoising achieves the best performance among the conventional graph signal denoising methods. The reason behind that is that the temperature data is smooth over the sensor network and the graph Laplacian-based prior nicely captures the smooth signal prior.

\mypar{NYC traffic data} 
We consider the taxi-pickup activity in Manhattan on January 1th, 2014. This is the Manhattan street network with $2,552$ intersections and $3,153$ road segments. We model each intersection as a vertex and each road segment as an edge. We model the taxi-pickup positions as signals supported on the Manhattan street network. We project each taxi-pickup to its nearest intersection, and then count the number of taxi-pickups at each intersection. Each graph signal is the hourly number of taxi-pickups recorded in each intersection; see one example in Fig.~\ref{fig:real_world_example} (b). We consider $24$ graph signals for one day in total. Compared to temperature data, the  graph signals here are much more complicated. Even two adjacent intersections have correlations, but they could have drastically different traffic behaviors.

We denoise two different numbers of graph signals: $1$ and $24$. Columns $4 - 5$ in Table~\ref{tab:real_world_denoise} compare the denoising performances under the mixture noise. Different from many other cases, all the denoising algorithms fail to achieve fine performances. The reason might be that this traffic data is too complicated. In this situation, we see that when more training data is available, the proposed two graph unrolling networks get better performance, reflecting the powerful learning ability to adapt to complicated data.

\mypar{Cora} 
We finally consider a citation network dataset, called Cora~\cite{KipfW:17}. The datasets contain sparse bag-of-words feature vectors for each document and a list of citation links between documents. We treat each citation link as an undirected edge and each document as a class label. The citation network has $2,708$ nodes and $5,429$ edges and $7$ class labels. We consider $7$ class labels as graph signals supported on this citation network. We introduce Bernoulli noises and randomly flip $10\%$ of the binary values.

We denoise two different numbers of graph signals: $1$ and $7$. Columns $6 - 9$ in Tables~\ref{tab:real_world_denoise} compare the denoising performances under Bernoulli noise with two evaluation metrics. 
For error rates, lower values mean better performances; for F1 scores, higher values mean better performances. We see that i) the proposed GUSC achieves the best performances in terms of both error rates and F1 scores; 2) since zeros appear much more frequently than ones in graph signals, most conventional methods tend to generate zero values everywhere, leading to good error rates, but bad F1 scores.

\section{Conclusions} 
We propose graph unrolling networks, which is an interpretable neural network framework to denoise single or multiple noisy graph signals. The proposed graph unrolling networks expand algorithm unrolling to the graph domain. As a core component of graph unrolling networks, we propose an edge-weight-sharing graph convolution operation, which parameterizes each edge weight by a trainable kernel function where the trainable  parameters are shared by all the edges. This convolution is permutation-equivariant and can flexibly adjust the edge weights to various graph signals. In the framework of graph unrolling networks, we propose two specific networks, graph unrolling sparse coding and graph unrolling trend filtering, by unrolling sparse coding and trend filtering, respectively. Through extensive experiments, we show that the proposed methods produce smaller denoising errors than both conventional denoising algorithms and state-of-the-art graph neural networks. Even for denoising a single graph signal, the normalized mean square error of the proposed networks is around $40\%$ and $60\%$ lower than that of graph Laplacian denoising and graph wavelets, respectively, reflecting the advantages of learning from only a few training samples.

\bibliographystyle{IEEEbib}
\bibliography{refs}
\end{document}